\documentclass[acmsmall,screen,nonacm]{acmart}

\AtBeginDocument{%
  }

\usepackage{amsmath}
\usepackage{amsthm}
\usepackage[T1]{fontenc}
\usepackage[utf8]{inputenc}
\usepackage[all,cmtip,2cell]{xy}
\usepackage{xcolor}
\usepackage{graphicx}
\usepackage{stmaryrd}
\usepackage{rotating}
\usepackage{enumitem}
\usepackage{algorithm}
\usepackage{algorithmic}
\usepackage{bussproofs}
\usepackage{bcprules}
\usepackage{mdframed}
\usepackage{prooftree}
\usepackage{cmll}
\usepackage[braket]{qcircuit}
\usepackage{pdflscape}
\usepackage{listings}
\usepackage{float}
\usepackage{cleveref}
\usepackage{subcaption}

\newtheorem{theorem}{Theorem}[section]
\newtheorem*{lemma*}{Lemma}
\newtheorem{corollary}[theorem]{Corollary}
\newtheorem{remark}[theorem]{Remark}
\newtheorem{definition}[theorem]{Definition}
\theoremstyle{definition}
\newtheorem{example}[theorem]{Example}

\allowdisplaybreaks

\usepackage{makecell}

\makeatletter
\renewcommand\subsubsection{\@startsection{subsubsection}{3}{\z@}%
  {3.25ex \@plus1ex \@minus.2ex}%
  {-1em}%
  {\normalfont\normalsize\bfseries}}

\newcounter{para}[subsubsection]

\usepackage{tikz}
\usetikzlibrary{arrows,snakes,automata,backgrounds, positioning,calc,cd}
\usetikzlibrary{shapes.geometric,shapes.misc,matrix,arrows.meta,decorations.markings}
\usetikzlibrary{circuits.logic.US, patterns, shapes}
\definecolor{comb}{HTML}{ddf1fb}
\definecolor{seq}{HTML}{6ae4a3}
\definecolor{neutral}{HTML}{f8f8f2}
\definecolor{unit}{HTML}{b9b9b9}

\newcommand{\dotsize}{3pt}
\tikzset{x=0.9em, y=2ex, baseline=-0.5ex}
\tikzset{ihbase/.style={inner sep=0,circle,draw,fill=lightgray,minimum size={\dotsize}, node contents={}}}
\tikzset{ihblack/.style={ihbase,fill=black}}
\tikzset{ihwhite/.style={ihbase,fill=white}}
\tikzset{mat/.style={draw,fill=white,rectangle,node font=\scriptsize}}
\tikzset{ha/.style={mat,rounded rectangle,rounded rectangle left arc=none}}
\tikzset{haop/.style={mat,rounded rectangle,rounded rectangle right arc=none}}
\tikzset{blackha/.style={mat,rounded rectangle,rounded rectangle left arc=none,font=\color{white},fill=black}}
\tikzset{blackhaop/.style={mat,rounded rectangle,rounded rectangle right arc=none,font=\color{white},fill=black}}
\tikzset{anti/.style={inner sep=0,isosceles triangle,fill=black,draw=black, minimum width=0.75em, node contents={}}}
\tikzset{antiop/.style={anti,shape border rotate=180}}
\tikzset{antisq/.style={inner sep=0,rectangle,fill=black, minimum height=1em, minimum width=0.6em, node contents={}}}
\tikzset{count/.style={above,inner ysep=0.15em,font=\scriptsize}}
\tikzset{axiom/.style={above,font=\small}}
\tikzset{dir/.style={-Latex}}
\tikzset{st/.style={decoration={markings,
    mark={at position 0.5 with {\draw (0, 2pt) to (0, -2pt);}}},
    postaction=decorate}}
\tikzstyle{none}=[inner sep=0mm]
\tikzstyle{flip}=[draw={rgb,255: red,86; green,86; blue,86}, fill={rgb,255: red,86; green,86; blue,86}, 
rounded rectangle, rounded rectangle right arc=none, text=white, node font={\scriptsize}, inner sep =0.2em]
\tikzstyle{cflip}=[circle, draw={rgb,255: red,128; green,128; blue,128}, fill={rgb,255: red,128; green,128; blue,128}, inner sep=0pt, minimum size=4pt]
\tikzstyle{and}=[fill=white, draw=black, and gate, scale=.4]
\tikzstyle{or}=[fill=white, draw=black, or gate, scale=.4, anchor=center]
\tikzstyle{not}=[fill=white, draw=black, not gate, scale=.35]
\tikzstyle{xor}=[fill=white, draw=black, xor gate, scale=.6]
\tikzstyle{if}=[trapezium, draw=black, fill=white, minimum width=6pt, minimum height=8pt, rotate=270]

\tikzstyle{plain}=[inner sep=0pt]
\tikzstyle{black}=[circle, draw=black, fill=black, inner sep=0pt, minimum size={\dotsize}]
\tikzstyle{black-faded}=[circle, draw=light-gray, fill=light-gray, inner sep=0pt, minimum size=4pt]
\tikzstyle{white}=[circle, draw=black, fill=white, inner sep=0pt, minimum size={\dotsize}]
\tikzstyle{white-faded}=[circle, draw=light-gray, fill=white, inner sep=0pt, minimum size=4.5pt]
\tikzstyle{reg}=[draw, fill=white, rounded rectangle, rounded rectangle left arc=none, minimum height=1.2em, minimum width=1.4em, node font={\scriptsize}]
\tikzstyle{effect}=[draw, fill=white, rounded rectangle, rounded rectangle left arc=none, minimum height=1.2em, minimum width=1.4em]
\tikzstyle{coreg}=[draw, fill=white, rounded rectangle, rounded rectangle right arc=none, minimum height=1.2em, minimum width=1.4em, node font={\scriptsize}]
\tikzstyle{box}=[shape=rectangle, text height=1.5ex, text depth=0.25ex, yshift=0.2mm, fill=white, draw=black, minimum height=3mm, minimum width=5mm, font={\small}]
\tikzstyle{basic box}=[shape=rectangle, text height=1.5ex, text depth=0.25ex, yshift=0.2mm, fill={probcircuitcolor}, draw={probcircuitcolor}, minimum height=3mm, minimum width=5mm, text={probcircuitcolorop}, inner sep=4pt]
\tikzstyle{small box}=[draw, fill={probcircuitcolor}, rectangle, minimum height=1.2em, minimum width=1.4em, node font={\scriptsize}, text={probcircuitcolorop}]
\definecolor{probcircuitcolor}{rgb}{0.33,0.33,0.33}
\definecolor{probcircuitcolorop}{rgb}{1.0,1.0,1.0}
\definecolor{boolcircuitcolor}{rgb}{1.0,1.0,1.0}
\definecolor{boolcircuitcolorop}{rgb}{0.0,0.0,0.0}
\tikzstyle{prob circuit}=[shape=rectangle, text height=1.5ex, text depth=0.25ex, yshift=0.2mm, fill={probcircuitcolor}, draw={probcircuitcolor}, minimum height=3mm, minimum width=5mm, text={probcircuitcolorop}, inner sep=4pt]
\tikzstyle{small prob circuit}=[draw={probcircuitcolor}, fill={probcircuitcolor}, rectangle, minimum height=1.2em, minimum width=1.4em, node font={\scriptsize}, text={probcircuitcolorop}]
\tikzstyle{tall prob circuit}=[shape=rectangle, text height=1.5ex, text depth=0.25ex, yshift=0.2mm, fill={probcircuitcolor}, draw={probcircuitcolor}, minimum height=8mm, minimum width=5mm, text={probcircuitcolorop}]
\tikzstyle{bool circuit}=[shape=rectangle, text height=1.5ex, text depth=0.25ex, yshift=0.2mm, fill={boolcircuitcolor}, draw=black, minimum height=3mm, minimum width=5mm, font={\small}, text={boolcircuitcolorop}]
\tikzstyle{small bool circuit}=[draw, fill={boolcircuitcolor}, rectangle, minimum height=1.2em, minimum width=1.4em, node font={\scriptsize},text={boolcircuitcolorop}]
\tikzstyle{tall bool circuit}=[rectangle, draw, fill={boolcircuitcolor}, minimum height=8mm, minimum width=5mm ,text={boolcircuitcolorop}]

\newcommand{\genericcomult}[2]{
  \begin{tikzpicture}
    \node at (1, 0) [ihbase,solid,name=copy,#1];
    \draw[#2] (copy) .. controls (1.25, 0.5) .. (2, 0.5);
    \draw[#2] (0, 0) -- (copy);
    \draw[#2] (copy) .. controls (1.25, -0.5) .. (2, -0.5);
  \end{tikzpicture}
}
\newcommand{\genericcomultn}[2]{
  \tikz {
    \draw (1, 0) node[ihbase,name=copy,#1] .. controls (1.25, 0.5) .. (1.5, 0.5)
    -- node[count] {$#2$} (2.25, 0.5);
    \draw (0, 0) -- node[count] {$#2$} (copy) .. controls (1.25, -0.5) .. (1.5, -0.5)
    -- node[count] {$#2$} (2.25, -0.5);
  }
}
\newcommand{\genericcounit}[2]{
  \tikz \draw[#2] (0, 0) -- (1, 0) node[ihbase,#1, solid];
}
\newcommand{\genericcounitn}[2]{
  \tikz \draw (0, 0) -- node[count] {$#2$} (1, 0) node[ihbase,#1];
}
\newcommand{\genericmult}[2]{
  \tikz {
    \node at (1,0) (copy) [ihbase,#1,solid];
    \draw[#2] (0,  0.5) .. controls (0.75,  0.5) .. (copy);
    \draw[#2] (0, -0.5) .. controls (0.75, -0.5) .. (copy);
    \draw[#2] (copy) -- (2, 0);
  }
}

\newcommand{\Bcomult}{\genericcomult{ihblack}{}}
\newcommand{\Bcomultn}[1]{\genericcomultn{ihblack}{#1}}
\newcommand{\Bcounit}{\genericcounit{ihblack}{}}
\newcommand{\Bcounitn}[1]{\genericcounitn{ihblack}{#1}}
\newcommand{\Bmult}{\genericmult{ihblack}{}}

\newcommand\idzero{
\InputIfFileExists{empty-diag.tikz}{}{\input{./tikz/empty-diag.tikz}}
} %

\newcommand{\idone}{
  \tikz \draw (0, 0) -- (2, 0);
}

\newcommand{\Andgate}{
\begin{tikzpicture}[circuit logic US]
	\begin{pgfonlayer}{nodelayer}
		\node [style=none] (91) at (-1.5, 0.5) {};
		\node [style=and] (93) at (0, 0) {};
		\node [style=none] (96) at (-1.5, -0.5) {};
		\node [style=none] (103) at (1.25, 0) {};
	\end{pgfonlayer}
	\begin{pgfonlayer}{edgelayer}
		\draw [in=0, out=150] (93) to (91.center);
		\draw [in=-150, out=0, looseness=0.75] (96.center) to (93);
		\draw (93) to (103.center);
	\end{pgfonlayer}
\end{tikzpicture}
}

\newcommand{\ConvexSum}[1]{
\begin{tikzpicture}
	\begin{pgfonlayer}{nodelayer}
		\node [style=none] (0) at (-3.25, 0) {};
		\node [style=none] (3) at (-0.75, 0) {};
		\node [style=if] (6) at (-0.5, 0) {};
		\node [style=none] (9) at (1, 0) {};
		\node [style=none] (10) at (-3.25, -0.5) {};
		\node [style=none] (11) at (-0.75, -0.5) {};
		\node [style=flip] (15) at (-2.5, 0.75) {$#1$};
		\node [style=none] (16) at (-0.75, 0.5) {};
	\end{pgfonlayer}
	\begin{pgfonlayer}{edgelayer}
		\draw (0.center) to (3.center);
		\draw (6) to (9.center);
		\draw [in=180, out=0] (10.center) to (11.center);
		\draw [in=180, out=0, looseness=1.50] (15) to (16.center);
	\end{pgfonlayer}
\end{tikzpicture}
}

\newcommand{\Orgate}{
\begin{tikzpicture}[circuit logic US]
	\begin{pgfonlayer}{nodelayer}
		\node [style=none] (91) at (-1.5, 0.5) {};
		\node [style=or] (93) at (0, 0) {};
		\node [style=none] (96) at (-1.5, -0.5) {};
		\node [style=none] (103) at (1.25, 0) {};
		\node [style=none] (104) at (-0.25, -0.25) {};
		\node [style=none] (105) at (-0.25, 0.25) {};
	\end{pgfonlayer}
	\begin{pgfonlayer}{edgelayer}
		\draw (93) to (103.center);
		\draw [in=180, out=0] (91.center) to (105.center);
		\draw [in=-180, out=0] (96.center) to (104.center);
	\end{pgfonlayer}
\end{tikzpicture}
}

\newcommand{\Notgate}{
\begin{tikzpicture}[circuit logic US]
	\begin{pgfonlayer}{nodelayer}
		\node [style=not] (93) at (0, 0) {};
		\node [style=none] (96) at (-1.25, 0) {};
		\node [style=none] (103) at (1.25, 0) {};
	\end{pgfonlayer}
	\begin{pgfonlayer}{edgelayer}
		\draw (96.center) to (93);
		\draw (93) to (103.center);
	\end{pgfonlayer}
\end{tikzpicture}
}

\newcommand{\Flip}[1]{
\begin{tikzpicture}
	\begin{pgfonlayer}{nodelayer}
		\node [style=flip] (93) at (0, 0) {$#1$};
		\node [style=none] (103) at (1.25, 0) {};
	\end{pgfonlayer}
	\begin{pgfonlayer}{edgelayer}
		\draw (93) to (103.center);
	\end{pgfonlayer}
\end{tikzpicture}
}

\newcommand{\Ifgate}{
\begin{tikzpicture}[circuit logic US]
	\begin{pgfonlayer}{nodelayer}
		\node [style=none] (91) at (-1.5, 0.5) {};
		\node [style=if] (93) at (0, 0) {};
		\node [style=none] (96) at (-1.5, 0) {};
		\node [style=none] (104) at (-0.25, 0) {};
		\node [style=none] (105) at (-0.25, 0.5) {};
		\node [style=none] (110) at (1.5, 0) {};
		\node [style=none] (111) at (-1.5, -0.5) {};
		\node [style=none] (112) at (-0.25, -0.5) {};
	\end{pgfonlayer}
	\begin{pgfonlayer}{edgelayer}
		\draw [in=180, out=0] (91.center) to (105.center);
		\draw [in=-180, out=0] (96.center) to (104.center);
		\draw (93) to (110.center);
		\draw [in=180, out=0] (111.center) to (112.center);
	\end{pgfonlayer}
\end{tikzpicture}
}

\newcommand{\idx}[1]{
  \tikz \draw (0, 0) -- (1.5, 0) node [midway, above] {\scriptsize $#1$};
}

\newcommand{\sym}{
  \tikz {
    \draw (0,  0.4) .. controls (0.5,  0.4) and (0.5, -0.4) .. (1, -0.4);
    \draw (0, -0.4) .. controls (0.5, -0.4) and (0.5,  0.4) .. (1,  0.4);
  }
}

\pgfdeclarelayer{edgelayer}
\pgfdeclarelayer{nodelayer}
\pgfsetlayers{edgelayer,nodelayer,main}

\definecolor{light-gray}{gray}{.5}

\tikzset{
BWmatrix/.pic={
    \coordinate (center) at (0,0);
    \filldraw[fill=white, draw=black, line width=1pt] (.5,0) 
        [rounded corners=14pt] -- (1,0) 
        [rounded corners=14pt] -- (1,1)
        [rounded corners=0pt] -- (.5,1) 
        [rounded corners=0pt] -- cycle;
    \filldraw[fill=black, draw=black, line width=1pt] (0,0) 
        -- (.5,0) 
        -- (.5,1)
        -- (0,1) 
        -- cycle;
   },
pics/BWmatrix/.default=0.2
}

\tikzstyle{arrow}=[->]

\newcommand{\probcircuit}[3]{
\begin{tikzpicture}
	\begin{pgfonlayer}{nodelayer}
		\node [style=prob circuit] (0) at (0, 0) {$#1$};
		\node [style=none] (1) at (1.75, 0) {};
		\node [style=none] (2) at (-1.75, 0) {};
		\node [style=none] (3) at (1.5, 0.5) {\scriptsize $#3$};
		\node [style=none] (4) at (-1.5, 0.5) {\scriptsize $#2$};
	\end{pgfonlayer}
	\begin{pgfonlayer}{edgelayer}
		\draw (2.center) to (0);
		\draw (0) to (1.center);
	\end{pgfonlayer}
\end{tikzpicture}
}

\newcommand{\smallprobcircuit}[3]{
\begin{tikzpicture}
	\begin{pgfonlayer}{nodelayer}
		\node [style=small prob circuit] (0) at (0, 0) {$#1$};
		\node [style=none] (1) at (1.75, 0) {};
		\node [style=none] (2) at (-1.75, 0) {};
		\node [style=none] (3) at (1.25, 0.5) {\scriptsize $#3$};
		\node [style=none] (4) at (-1.25, 0.5) {\scriptsize $#2$};
	\end{pgfonlayer}
	\begin{pgfonlayer}{edgelayer}
		\draw (2.center) to (0);
		\draw (0) to (1.center);
	\end{pgfonlayer}
\end{tikzpicture}
}

\newcommand{\boolcircuit}[3]{
\begin{tikzpicture}
	\begin{pgfonlayer}{nodelayer}
		\node [style=bool circuit] (0) at (0, 0) {$#1$};
		\node [style=none] (1) at (1.75, 0) {};
		\node [style=none] (2) at (-1.75, 0) {};
		\node [style=none] (3) at (1.5, 0.5) {\scriptsize $#3$};
		\node [style=none] (4) at (-1.5, 0.5) {\scriptsize $#2$};
	\end{pgfonlayer}
	\begin{pgfonlayer}{edgelayer}
		\draw (2.center) to (0);
		\draw (0) to (1.center);
	\end{pgfonlayer}
\end{tikzpicture}
}

\newcommand{\diagbox}[3]{
\begin{tikzpicture}
	\begin{pgfonlayer}{nodelayer}
		\node [style=basic box] (0) at (0, 0) {$#1$};
		\node [style=none] (1) at (1.5, 0) {};
		\node [style=none] (2) at (-1.5, 0) {};
		\node [style=none] (3) at (1.5, 0.5) {\scriptsize $#3$};
		\node [style=none] (4) at (-1.5, 0.5) {\scriptsize $#2$};
	\end{pgfonlayer}
	\begin{pgfonlayer}{edgelayer}
		\draw (2.center) to (0);
		\draw (0) to (1.center);
	\end{pgfonlayer}
\end{tikzpicture}
}

\newcommand{\myeq}[1]{\mathrel{\overset{\makebox[0pt]{\mbox{\normalfont\tiny\sffamily #1}}}{=}}}

\def\distto{\mathrel{\mkern3mu  \vcenter{\hbox{$\scriptscriptstyle+$}}%
                    \mkern-12mu{\to}}}

\newcommand{\distcomp}{\,;}
\newcommand{\Prob}{\mathbb{P}}

\newcommand{\sem}[1]{\left\llbracket{#1}\right\rrbracket}
\newcommand{\condsem}[1]{\left\llbracket{#1}\right\rrbracket_\propto}
\newcommand{\diagsem}[1]{|#1|}

\newcommand{\Bend}[1]{#1^\flat}
\newcommand{\Bernoulli}[1]{\mathtt{flip}\;{#1}}
\newcommand{\Let}[3]{\texttt{let }#1 = #2 \texttt{ in } #3}
\newcommand{\Bayes}[2]{{#1}^\dagger_{#2}}
\newcommand{\All}[1]{\mathsf{all}_{#1}}

\newcommand{\Dist}{\mathcal{D}}
\newcommand{\SubDist}{\mathcal{D}_{\leq 1}}

\newcommand{\N}{\mathbb{N}}

\newcommand{\Bool}{\mathbb{B}}

\newcommand{\ProbCirc}{\mathsf{ProbCirc}}
\newcommand{\BoolCirc}{\mathsf{BoolCirc}}
\newcommand{\CausCirc}{\mathsf{CausCirc}}

\newcommand{\id}{\mathrm{id}}

\newcommand{\from}{\mathrel{:}\,}

\newcommand{\poi}{\,;\,}
\newcommand{\tns}{\otimes}
\newcommand{\adjto}{\,\lower1pt\hbox{$\dashv$}\,}

\newcommand{\Set}{\mathsf{Set}}

\newcommand{\Kl}[1]{{{Kl}}(#1)}

\newcommand{\fStoch}{\mathsf{FinStoch}}
\newcommand{\fProjStoch}{\mathsf{FinProjStoch}}
\newcommand{\fSubStoch}{\mathsf{FinSubStoch}}

\makeatletter
\def\moverlay{\mathpalette\mov@rlay}
\def\mov@rlay#1#2{\leavevmode\vtop{%
\baselineskip\z@skip \lineskiplimit-\maxdimen
\ialign{\hfil$#1##$\hfil\cr#2\crcr}}}
\makeatother

\newcommand{\typ}{\mathrel{:}}
\newcommand{\typeJudgment}[3]{{ {#2} \,\typ\, {#3}}}
\newcommand{\sort}[2]{\ensuremath{#1 \to #2}}

\newcommand{\reductionRule}[2]{{\prooftree{\scriptstyle #1}\justifies{\scriptstyle #2}\endprooftree}}

\newcommand\twarr[2]{%
\mathrel{\mathop{\moverlay{\scriptstyle\xrightarrow{\,#1\,}\cr{\lower.2em\hbox{$\scriptstyle{}_{#2}$}}}}}}
\newcommand\twarrw[2]{%
\mathrel{\mathop{\moverlay{\scriptstyle\Longrightarrow\cr{\lower-.6em\hbox{$\scriptstyle{}_{#1}$}}
\cr{\lower.3em\hbox{$\scriptstyle{}_{#2}$}}}}}}

\newcommand{\dtransw}[2]{\raise1pt\hbox{$\;\twarrw{#1}{#2}\;$}}

\newcommand{\diagregexp}[1]{
\begin{tikzpicture}
	\begin{pgfonlayer}{nodelayer}
		\node [style=none] (0) at (1.5, 0) {};
		\node [style=rcoreg] (1) at (0, 0) {{\color{red} $e$}};
	\end{pgfonlayer}
	\begin{pgfonlayer}{edgelayer}
		\draw [red] (1) to (0.center);
	\end{pgfonlayer}
\end{tikzpicture}}

\begin{document}

\title{A Complete Axiomatisation of Equivalence for Discrete Probabilistic Programming}

\author{Robin Piedeleu}
\authornote{Both authors contributed equally to this research.}
\email{r.piedeleu@ucl.ac.uk}
\affiliation{%
  \institution{University College London}
  \country{United Kingdom}
}
\author{Mateo Torres-Ruiz}
\authornotemark[1]
\email{m.torresruiz@cs.ucl.ac.uk}
\affiliation{%
  \institution{University College London}
  \country{United Kingdom}
}

\author{Alexandra Silva}
\affiliation{%
 \institution{Cornell University}
 \city{Ithaca}
 \state{NY}
 \country{USA}}
\email{alexandra.silva@cornell.edu}

\author{Fabio Zanasi}
\affiliation{%
  \institution{University College London}
  \city{London}
  \country{United Kingdom}}
\affiliation{%
  \institution{University of Bologna}
  \country{Italy}}
\email{f.zanasi@ucl.ac.uk}

\renewcommand{\shortauthors}{Piedeleu, Torres-Ruiz, Silva, Zanasi}

\begin{abstract}
We introduce a sound and complete equational theory capturing equivalence of discrete probabilistic programs, that is, programs extended with primitives for Bernoulli distributions and conditioning, to model distributions over finite sets of events. To do so, we translate these programs into a graphical syntax of probabilistic circuits, formalised as string diagrams, the two-dimensional syntax of symmetric monoidal categories. We then prove a first completeness result for the equational theory of the conditioning-free fragment of our syntax. Finally, we extend this result to a complete equational theory for the entire language. Note these developments are also of interest for the development of probability theory in Markov categories: our first result gives a presentation by generators and equations of the category of Markov kernels, restricted to objects that are powers of the two-elements set.
\end{abstract}

\ccsdesc[500]{Theory of computation~Logic}
\ccsdesc[500]{Theory of computation~Categorical semantics}
\keywords{probabilistic circuits, string diagrams, complete axiomatisation}

\maketitle

\section{Introduction}
\label{sec:introduction}

Probabilistic programming languages (PPLs) extend standard languages with two added capabilities: drawing random values from a given distribution, and conditioning on particular variable values through observations. 
These constructs allow the programmer to define complex statistical models and obtain the probability of an event according to the distribution specified by the program, a task known as \emph{inference}. This process can also be understood in Bayesian terms: first-class distributions define a prior over the program's declared variables, while observations record the likelihood of specific values. Inference then involves computing the posterior distribution over some chosen variables.

In this paper, we focus on \emph{discrete} probabilistic programs, that is, programs whose associated probability distributions range over a finite set of outcomes. While many existing PPLs manipulate continuous random variables, this generally forces them to approximate the posterior distribution specified by a given program, typically by sampling from it \cite{anglicanpaper,hakarupaper,lazypplpaper}. Discrete probabilistic programs on the other hand lend themselves to \emph{exact} inference, where instead of approximating or sampling from the posterior distribution, it is possible to obtain an exact representation of it. 

Several existing languages support this style of inference for discrete probabilistic programs~\cite{dicepaper,bernoulliprobpaper,psipaper}. In this work, we start from a prototypical discrete PPL that incorporates a primitive for Bernoulli distributions (\texttt{flip}) and observations (\texttt{observe}). Beyond its probabilistic features, its syntax is that of a simple first-order, non-recursive functional programming language with \texttt{let}-expressions for variable declaration and branching with \texttt{if-then-else}. In this respect, it is similar to \textit{Dice}~\cite{dicepaper} or the \emph{CD-calculus}~\cite{stein2023probabilistic}. While it may appear limited, it is in fact sufficiently expressive to specify arbitrary distributions over tuples of Booleans and thus, to encode arbitrary discrete distributions.
The following examples illustrate how to express simple models as probabilistic programs. 
\begin{example}
 Assume that we have an opaque urn that contains one ball which we cannot see and is equally likely to be either red or blue ($\texttt{firstball } = \texttt{flip }0.5$). We place a red ball into the urn ($\texttt{redball} = \texttt{true}$), and draw at random one of the two balls now contained within it ($\texttt{draw} = \texttt{if flip }0.5 \texttt{ then redball else firstball}$). We are interested in the following inference problem: given that the drawn ball is red ($\texttt{observe draw}$), what is the probability that the first ball in the urn was also red ($\dots\texttt{in firstball}$)? This puzzle can be formalised by the following program, where \texttt{true} encodes a ball being red and \texttt{false} encodes a ball being blue:
  \begin{lstlisting}[xleftmargin=.2\textwidth,basicstyle=\ttfamily\small, escapeinside={<@}{@>}]
let firstball = <@\color{violet}flip 0.5@> in 
let redball = <@\color{red}true@> in 
let draw = <@\color{blue}if flip 0.5 then redball else firstball@> in 
let _ = <@\color{orange}observe draw@> in firstball
  \end{lstlisting}
\end{example}
\begin{example}\label{ex:von-neumann}
Imagine that you are betting on the outcome of a coin flip and that you cannot trust that the coin is fair. You can always agree to use the following protocol, due to John Von Neumann, to simulate a fair coin: toss the coin twice; if the results are the same, discard them and start over; if they are different, keep the first outcome and discard the second. This process can be translated as the following program, where $p$ is an arbitrary parameter strictly between $0$ and $1$, representing the unknown probability that the coin lands on heads.
  \begin{lstlisting}[xleftmargin=.2\textwidth,basicstyle=\ttfamily\small, escapeinside={<@}{@>}]
let first = <@\color{violet}flip p@> in 
let second = <@\color{red}flip p@> in 
let compare = <@\color{blue}first xor second@> in 
let _ = <@\color{orange}observe compare@> in first
  \end{lstlisting}
\end{example}
As these examples illustrate, probabilistic programs are a compact way to write and communicate probabilistic models: they are modular, compositional, and understood by a large community of practitioners, who can then leverage functional abstraction to build and share increasingly complicated models.

A probabilistic program can often be replaced by an equivalent program whose corresponding inference task is simpler to solve. %
Here, we say that two probabilistic programs are \emph{equivalent} when they define the same distribution.
The Von Neumann trick to simulate a fair coin from a biased one gives a first example of a non-trivial program equivalence: the program of Example~\ref{ex:von-neumann} is equivalent to the much simpler program $\texttt{flip }0.5$. We give another example of equivalence below, for two \emph{partial} programs, \emph{i.e.}, programs containing free variables.
\begin{example}
\label{ex:equivalence-frob}
The partial program below, with free variable \texttt{x}, is equivalent to the identity function which simply returns \texttt{x}:
  \begin{lstlisting}[xleftmargin=.2\textwidth,basicstyle=\ttfamily\small, escapeinside={<@}{@>}]
let y = <@\color{black}flip 0.5@> in  
let compare = <@\color{black}x xor y@> in 
let _ = <@\color{black}observe (not compare)@> in y
  \end{lstlisting}
\end{example}
Program equivalence is essential to ensure that the behaviour of programs remain consistent across different implementations, optimisations, and refinements~\cite{katpaper,lawsofprogramming}. In probabilistic programming, as in standard programming, equivalence checking is crucial for building reliable, efficient, and accurate models, as well as inference algorithms. However, probabilistic features greatly complicate the task of verifying program equivalence.

\paragraph{Methodology} We tackle the problem of program equivalence by presenting a complete equational theory for probabilistic programs. To do so, we translate the simple PPL syntax above into an equally expressive diagrammatic calculus of \textit{probabilistic circuits}.  

We start from plain Boolean circuits with standard logical gates, and extend them with two additional gates that incorporate the probabilistic features of PPLs: one allowing us to draw from a Bernoulli distribution; %
 another to condition on two variables having the same value. 

While circuits are usually treated as graphs and used as informal visual aids, the circuits in this paper are given as a formal syntax of \emph{string diagrams}, the graphical language of (symmetric) monoidal categories~\cite{selinger2010}. In recent years, string diagrams have found their way into diverse scientific fields:
well-known examples include digital circuits in quantum computing~\cite{quantumprocessesbook}, electrical engineering \cite{electriccircuits}, control engineering \cite{categoriescontrol}, and more \cite{petrinets,bayesiannetworks}. %

Formalising circuits as string diagrams has several benefits. First, it allows us to reason about circuits as bona-fide algebraic objects with operations of sequential and parallel compositions. Moreover, we are able to equip them with a formal semantics defined inductively over the diagrammatic syntax, assigning to each circuit the probability distribution it is intended to represent. As we will see, this defines a \emph{compositional} interpretation automatically. 

Finally, we
define an equational theory which fully characterises semantic equality of circuits, the main result of our paper. In other words, we prove that whenever two circuits denote the same distribution, we can derive this fact using purely equational reasoning on the circuits themselves. For this, we proceed in two steps: first, we give a complete equational theory for the conditioning-free fragment of our syntax; second, we extend the previous equational theory to a complete equational theory for the whole language.

Note that these results can be seen as part of a broader program aiming to formulate probability theory in \emph{Markov categories}~\cite{fritzmarkovcats}. In particular, our first main result can also be seen as a presentation in terms of generators and equations of the full monoidal subcategory of the category of Markov kernels/stochastic maps ~\cite{giry1982categorical} on objects that are powers of the two-element set. This presentation is thus one further step on the way to developing an axiomatic perspective on probability theory in Markov categories.

\paragraph{Outline and main contributions} 
In Section~\ref{sec:syntax-semantics}, we give the formal syntax of probabilistic circuits with which we will work throughout the paper, both as a standard term syntax and as a two-dimensional syntax of string diagrams; we explain its relation with a conventional language for discrete probabilistic programming such as the one we have used in the previous examples; finally, we equip it with a compositional semantics which captures the probability distribution the circuits are intended to represent. In Section~\ref{sec:axiomatisation-caus-circ} we introduce and prove the completeness of an equational theory for the conditioning-free fragment of our circuits. In Section~\ref{sec:observe} we extend this equational theory to the full syntax, and prove it complete, thus obtaining a full axiomatisation of discrete probabilistic program equivalence. In conclusion, we compare our contributions with related work and explore some directions for future work.

\section{Syntax and Semantics}
\label{sec:syntax-semantics}

\subsection{Syntax}
\label{sec:syntax}

To define our syntax, we will proceed in two steps: we will start with a simple grammar whose constants are circuit primitives with rules for their sequential and parallel composition, before moving to a two-dimensional representation of the same circuits as string diagrams. Finally, we will compare our graphical language with a more conventional PL syntax, and give a translation of the latter into the former.

\subsubsection{Term language.}
\label{sec:term-syntax} 
The syntax $\ProbCirc$ is given by the following simple grammar:
\begin{align}
c \; ::= &\; \Bcounit \mid \Bcomult \mid \Andgate \mid \Notgate
 \mid \idzero \mid \idone \mid \sym   \mid c\poi c \mid c \tns c  \label{eq:bool-syntax}\\
& \mid \Flip{p} \quad\forall p\in [0,1] \label{eq:causal-syntax}\\
& \mid \Bmult \label{eq:full-syntax}
\end{align}
Though the constants of our language are depicted as diagrams, %
we treat them as symbols for the moment. Soon, we will consider all $\ProbCirc$-terms as \emph{string diagrams} and view the two binary operations of \emph{sequential} ($c \poi d$) and \emph{parallel} ($c \tns d$) composition as those of a monoidal category~\cite{selinger2010}.

We will refer to the syntax which comprises only the rules and constants of line~\eqref{eq:bool-syntax} as $\BoolCirc$, and the syntax comprising the rules and constants of line~\eqref{eq:bool-syntax} and~\eqref{eq:causal-syntax} as $\CausCirc$.

Note that our language does not use variables, nor do we need to define alpha-equivalence or substitution. On the other hand, we do need simple typing rules: a type is a pair $(m,n)\in \N\times \N$ which we write as $\sort{m}{n}$; we write the judgment that $c$ has type $\sort{m}{n}$ as $c\colon\sort{m}{n}$. From now on, we will only consider typeable terms, according to the rules of Fig.~\ref{fig:typing-rules}.
\begin{figure}
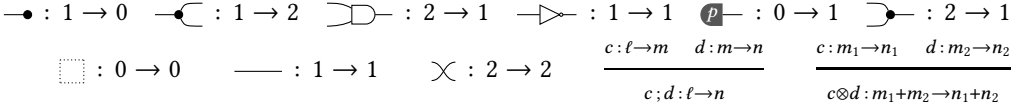

\begin{align*}
\typeJudgment{}{\Bcounit}{\sort{1}{0}} \quad
\typeJudgment{}{\Bcomult}{\sort{1}{2}} \quad
\typeJudgment{}{\Andgate}{\sort{2}{1}} \quad
\typeJudgment{}{\Notgate}{\sort{1}{1}} \quad
 \typeJudgment{}{\Flip{p}}{\sort{0}{1}} \quad
\typeJudgment{}{\Bmult}{\sort{2}{1}}
\\
 \typeJudgment{}{\idzero}{\sort{0}{0}} \qquad
 \typeJudgment{}{\idone}{\sort{1}{1}} \qquad
 \typeJudgment{}{\sym}{\sort{2}{2}} \qquad
\reductionRule{ \typeJudgment{}{c}{\sort{\ell}{m}} \quad \typeJudgment{}{d}{\sort{m}{n}} }
{ \typeJudgment{}{c\poi d}{\sort{\ell}{n}} }\qquad
\reductionRule{ \typeJudgment{}{c}{\sort{m_1}{n_1}} \quad \typeJudgment{}{d}{\sort{m_2}{n_2}} }
{ \typeJudgment{}{c \tns d}{\sort{m_1+m_2}{n_1+n_2}} }
\end{align*}
\caption{Typing rules for $\ProbCirc$ terms.}\label{fig:typing-rules}
\end{figure}
The type of a term $c$ corresponds to the number of open wires on each side of the diagram depicting $c$. A simple induction confirms the uniqueness of types: if $c \colon \sort{m}{n}$ and $c \colon \sort{m'}{n'}$, then $m = m'$ and $n = n'$. In particular, a circuit with no input (\emph{resp.} output) wires has type $0\to n$ (\emph{resp.} $m\to 0$). 

Just like traditional Boolean circuits, our terms are intended to represent networks of logical gates connected by wires carrying Boolean values. 
The purely Boolean fragment of the syntax ($\BoolCirc$) includes standard gates for logical conjunction, $\Andgate$, which outputs a 1 only when its two inputs are set to 1; logical negation, $\Notgate$, which negates its input; a broadcasting gate, $\Bcomult$, that copies its input to two output wires; and a terminating wire, $\Bcounit$, that discards its input. Beyond Boolean circuits, there are two additions. Firstly, $\CausCirc$ introduces $\Flip{p}$ which is the only probabilistic component of our syntax, emitting 1 with probability $p$ and 0 with probability $1-p$. Note that $\Flip{p}$ has \emph{no input} and that the parameter $p$ is simply a label which is not accessible to the rest of the circuit. Thus we have a generator for each $p\in[0,1]$.  Finally, $\ProbCirc$ adds a component for explicit conditioning, $\Bmult$, which constrains its two input wires to carry equal values and outputs this value. A similar primitive has appeared in previous work~\cite{stein2023probabilistic}. As we will see, this is the only component whose behaviour is not causal, in the sense that it constrains the values of its input. This is why we call \emph{causal} circuits the conditioning-free fragment of our syntax ($\CausCirc$).  %
Intuitively, these circuits specify a distribution over all possible values of their output wires for any value of their input wires. We will assign them a formal semantics in Section~\ref{sec:semantics}.

\begin{example}\label{ex:terms}
From our generators, we can define other standard Boolean gates as syntactic sugar, using sequential and parallel composition. We will need \texttt{OR}-gates, defined as
\[\Orgate := (\Notgate \tns \Notgate)\poi \Andgate \poi \Notgate\]
and multiplexers/\texttt{if-then-else} gates as:
\[\Ifgate :=(\Bcomult \tns \idone \tns \idone)\poi (\Notgate \tns \sym\tns \idone)\poi (\Andgate \tns \Andgate)\poi \Orgate\]
\end{example}

\subsubsection{From terms to string diagrams.}
\label{sec:term-to-diagrams}
As is clear from the previous examples, $\ProbCirc$-terms are not easy to parse for the human eye and come with lots of redundant information about how a given term has been put together. Moreover, our main goal is to obtain an equational axiomatisation of semantic equivalence of $\ProbCirc$-terms; when we define their semantics formally, we will see that their interpretation satisfies all the axioms of \emph{symmetric monoidal categories} (SMCs). This licences us to move from terms to \emph{string diagrams}, a common graphical representation of morphisms in SMCs~\cite{selinger2010,maclane1971}. We recall some of the basics below, though we refer the reader to a more suitable introduction for more information on the topic~\cite{introdiagrams}.
Formally, a string diagram on $\ProbCirc$ is defined as an equivalence class of $\ProbCirc$-terms, where the quotient is taken with respect to the reflexive, symmetric and transitive closure of the following axioms:
\[(c\poi d)\poi e = c\poi (d\poi e) \quad c_1\otimes (c_2\otimes c_3) = (c_1\otimes c_2)\otimes c_3\quad  (c_1\otimes c_2)\poi (d_1\otimes d_2)=(c_1\poi d_1)\otimes (c_2\poi d_2)\] 
\[c\poi \idone = c = \idone \poi c\qquad \idzero\otimes c =c= c\otimes \idzero\]
\[(\idone\otimes c)\poi\sym = \sym\poi (c\otimes \idone)\qquad\sym\poi \sym = \idone\otimes \idone\]
where $c, d, e$ and $c_i$, $d_i$ range over $\ProbCirc$-terms of the appropriate type. These axioms state that sequential and parallel composition are associative and unital, that they satisfy a form of interchange law, and that the wire crossings behave as our topological intuition would expect---see their diagrammatic translation~\eqref{eq:smc} below.

Let us establish some conventions. We will refer to string diagrams representing $\BoolCirc$-terms as \emph{Boolean circuits}, those representing $\CausCirc$-terms as \emph{causal circuits}, and arbitrary string diagrams for $\ProbCirc$-terms simply as \emph{circuits}.
We will depict $n$ parallel wires, for some natural number $n$, by a labelled-wire $\idx{n}$ (and omit the label when $n=1$). Moreover, we will depict a generic circuit $c\from m\to n$ by a labelled black box, and a Boolean circuit $b\from m\to n$ by a labelled white box:
\[\probcircuit{c}{m}{n} = 
\InputIfFileExists{generic-circuit-m-n.tikz}{}{\input{./tikz/generic-circuit-m-n.tikz}}
\qquad\qquad \boolcircuit{b}{m}{n} = 
\InputIfFileExists{boolean-circuit-m-n.tikz}{}{\input{./tikz/boolean-circuit-m-n.tikz}}
\]
Then, we represent $c\poi d$ as sequential composition, depicted horizontally from left to right (note that the types of the intermediate wires have to match), and $c\tns d$ as parallel composition, depicted vertically from top to bottom:
\[\probcircuit{\scriptstyle{c\poi d}}{\ell\;}{\; n} \;=\;
\InputIfFileExists{seq-compose.tikz}{}{\input{./tikz/seq-compose.tikz}}

\qquad 
\InputIfFileExists{c1xc2.tikz}{}{\input{./tikz/c1xc2.tikz}}
 \;=\;
\InputIfFileExists{par-compose.tikz}{}{\input{./tikz/par-compose.tikz}}
\]
We omit the labels on the wires for readability when they can be inferred unambiguously from the context.

Intuitively, our string diagrams can be formed much like conventional circuits, by wiring the generators of~\eqref{eq:bool-syntax}-\eqref{eq:full-syntax} in sequence or in parallel, crossing wires (with $\sym$) and making wires as long as we want (with $\idone$). As string diagrams, the axioms of SMCs now become near-tautologies:
\begin{equation}\label{eq:smc}
\begin{gathered}
\!\!\!\!\!{
\InputIfFileExists{smc/sequential-associativity.tikz}{}{\input{./tikz/smc/sequential-associativity.tikz}}
 \!=\! 
\InputIfFileExists{smc/sequential-associativity-1.tikz}{}{\input{./tikz/smc/sequential-associativity-1.tikz}}
} \quad \scalebox{1}{
\InputIfFileExists{smc/parallel-associativity.tikz}{}{\input{./tikz/smc/parallel-associativity.tikz}}
 \!\!\!=\!\!\! 
\InputIfFileExists{smc/parallel-associativity-1.tikz}{}{\input{./tikz/smc/parallel-associativity-1.tikz}}
}\quad  \scalebox{1}{
\InputIfFileExists{smc/interchange-law.tikz}{}{\input{./tikz/smc/interchange-law.tikz}}
 \!\!=\!\!
\InputIfFileExists{smc/interchange-law-1.tikz}{}{\input{./tikz/smc/interchange-law-1.tikz}}
 }
 \\
\scalebox{1}{
\InputIfFileExists{smc/unit-right.tikz}{}{\input{./tikz/smc/unit-right.tikz}}
 = \diagbox{c}{}{} = 
\InputIfFileExists{smc/unit-left.tikz}{}{\input{./tikz/smc/unit-left.tikz}}
}
\quad
  \scalebox{1}{ 
\InputIfFileExists{smc/parallel-unit-above.tikz}{}{\input{./tikz/smc/parallel-unit-above.tikz}}
 = \diagbox{c}{}{} =  
\InputIfFileExists{smc/parallel-unit-below.tikz}{}{\input{./tikz/smc/parallel-unit-below.tikz}}
}
\\
\scalebox{1}{
\InputIfFileExists{smc/sym-natural.tikz}{}{\input{./tikz/smc/sym-natural.tikz}}
= 
\InputIfFileExists{smc/sym-natural-1.tikz}{}{\input{./tikz/smc/sym-natural-1.tikz}}
}
\qquad
\scalebox{1}{
\InputIfFileExists{smc/sym-iso.tikz}{}{\input{./tikz/smc/sym-iso.tikz}}
 = 
\begin{tikzpicture}
	\begin{pgfonlayer}{nodelayer}
		\node [style=none] (0) at (2, -0.75) {};
		\node [style=none] (1) at (-2, -0.75) {};
		\node [style=none] (2) at (-2, 0.5) {};
		\node [style=none] (3) at (2, 0.5) {};
	\end{pgfonlayer}
	\begin{pgfonlayer}{edgelayer}
		\draw (0.center) to (1.center);
		\draw (3.center) to (2.center);
	\end{pgfonlayer}
\end{tikzpicture}}
}
\end{gathered}
\end{equation}
If we think of the dotted frames as two-dimensional brackets, these identities tell us that the specific bracketing of a term does not matter. This is precisely the advantage of working with string diagrams, rather than terms: they free us from some of the bureaucracy of terms, allowing us to focus on the more structural aspects of our syntax, \emph{i.e.} on how the different components making up a term are wired together. This means in particular that wire crossings obey laws that are topologically obvious, as shown in~\eqref{eq:smc}. Finally, string diagrams give us the best of both worlds, as they remain an inductively-specified structure on which we can reason by induction, unlike the monolithic representation of traditional circuits as graphs. 

\begin{example}\label{eq:or-if-def}
Coming back to the \texttt{OR}-gate and multiplexers defined in Example~\ref{ex:terms}, we have
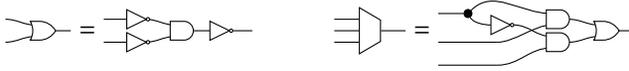
\begin{equation*}
\Orgate = 
\InputIfFileExists{or-def.tikz}{}{\input{./tikz/or-def.tikz}}
\quad\qquad  \Ifgate = 
\InputIfFileExists{if-def.tikz}{}{\input{./tikz/if-def.tikz}}

\end{equation*}
Even though we have not defined the formal semantics yet, their interpretation is already clearer from their depiction as circuits: the first will behave like a standard logical \texttt{OR}-gate, which outputs the disjunction of its two inputs, and the second like a multiplexer or \texttt{if-then-else}-gate, which outputs either its second input if the first is set to $1$ or the third input if the first is set to $0$. We will call the first input the \emph{guard}, and the next two, the \texttt{then}- and \texttt{else}-\emph{branch}, respectively.

We also define generalised gates for $n$ wires, as syntactic sugar, by induction: for \texttt{AND}-gates and copying nodes, let
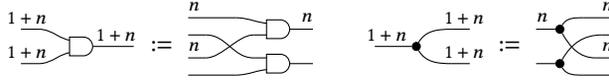
\begin{equation*}
\label{eq:general-gates}

\InputIfFileExists{and-1xn.tikz}{}{\input{./tikz/and-1xn.tikz}}
\; := \;
\InputIfFileExists{and-1xn-def.tikz}{}{\input{./tikz/and-1xn-def.tikz}}
\qquad 
\InputIfFileExists{copy-1xn.tikz}{}{\input{./tikz/copy-1xn.tikz}}
 \;:=\;
\InputIfFileExists{copy-1xn-def.tikz}{}{\input{./tikz/copy-1xn-def.tikz}}

\end{equation*}
with base cases the corresponding generators; we define generalised gates for $\Bcounit, \Notgate$, $\Bmult$, etc. in a similar way.
\end{example}

\begin{example}
  Consider a simple setting where we flip a coin $x$ with a $0.1$ chance of landing heads. If $x$ is heads, we flip a new coin $y_1$ with a $0.2$ chance of heads; otherwise, we flip a third coin $y_2$ with a $0.3$ chance of heads. Finally, based on the outcome of $y_i$ (for either $i=1$ or $i=2$), we choose between two more coins $z_1$ or $z_2$: if $y_i$ is heads, $z_1$ has a $0.4$ chance of heads; otherwise, we flip $z_2$, which is a fair coin. Now, imagine that we are interested in the probability of the last coin flip landing on heads. This is what the following causal circuit models:
  \begin{center}
    
\InputIfFileExists{example-01.tikz}{}{\input{./tikz/example-01.tikz}}

  \end{center}
  \label{fig:example-01}
\end{example}

\subsubsection{Comparing with a conventional PL syntax.}
\label{sec:compare-with-PPL}
While our syntax is purposefully rendered graphically, 
there are existing PLs specialised in discrete probabilistic programming that use a more conventional syntax.
Fig.~\ref{fig:typing-rules} gives the syntax and typing rules of  
a simple (first-order) functional probabilistic language,
which is as expressive as our circuits: the judgment $\Gamma \vdash e\colon \tau$ denotes a well-typed partial program $e$ in context $\Gamma$, which is a list of pairs $x:\tau$ of free variables with their types, and types are simply tuples of Booleans. The language's syntax mimics closely that of a probabilistic PL called \textit{Dice}~\cite{dicepaper}. A similar language also appears in~\cite[Section 6.3]{stein2023probabilistic} under the name of the \emph{CD-calculus}---indeed, it is a natural domain-specific language for discrete probabilistic programming. Let us compare it with our diagrammatic calculus.

\begin{figure}
  \begin{equation*}
\begin{minipage}{.2\linewidth}
      \centering
      \infrule{}{\Gamma, x:\tau,\Delta\vdash x:\tau}
    \end{minipage} 
    \begin{minipage}{.4\linewidth}
      \centering
      \infrule{\Gamma\vdash e_1: \tau_1 \andalso \Gamma\vdash e_2:\tau_2}
              {\Gamma\vdash (e_1,e_2):\tau_1\times \tau_2}
    \end{minipage} 
        \begin{minipage}{.2\linewidth}
      \centering
      \infrule{\Gamma\vdash e:\tau_1\times\tau_2}{\Gamma\vdash\mathtt{fst}\;e:\tau_1}
    \end{minipage}
        \begin{minipage}{.2\linewidth}
      \centering
      \infrule{\Gamma\vdash e:\tau_1\times\tau_2}{\Gamma\vdash\mathtt{snd}\;e:\tau_2}
    \end{minipage}       
\end{equation*} 
\begin{equation*}
    \begin{minipage}{0.4\linewidth}
      \centering
      \infrule{ \Gamma\vdash g:\Bool\quad\Gamma\vdash e_0:\tau\quad \Gamma\vdash e_1:\tau}{\Gamma\vdash\mathtt{if}\;g\;\mathtt{then}\;e_1\;\mathtt{else}\;e_0:\tau}
    \end{minipage}
        \begin{minipage}{.2\linewidth}
      \centering
      \infrule{}{\Gamma\vdash\Bernoulli{p}:\Bool}
    \end{minipage}
      \begin{minipage}{.2\linewidth}
      \centering
      \infrule{}{\Gamma\vdash\mathtt{true}:\Bool}
    \end{minipage}
      \begin{minipage}{.2\linewidth}
      \centering
      \infrule{}{\Gamma\vdash\mathtt{false}:\Bool}
    \end{minipage}
\end{equation*}
\begin{equation*}
    \begin{minipage}{.5\linewidth}
      \centering
      \infrule{\Gamma\vdash e_0: \tau_0 \andalso \Gamma,x:\tau_0\vdash e_1:\tau}
              {\Gamma\vdash\mathtt{let}\;x=e_0\;\mathtt{in}\;e_1:\tau}
    \end{minipage} 
    \begin{minipage}{.5\linewidth}
      \centering
      \infrule{}
              {\Gamma,x:\tau,\Delta\vdash\mathtt{observe}\; x\;:\tau}
    \end{minipage} 
\end{equation*}
\begin{equation*}
    \begin{minipage}{0.5\linewidth}
      \centering
      \infrule{ \Gamma,x:\tau,\Delta \vdash e':\tau'\quad (f \text{ fresh})}{\Gamma,x:\tau,\Delta\vdash\mathtt{fun}\, f(x:\tau)\, \{e'\}:\tau'}
    \end{minipage}
        \begin{minipage}{0.5\linewidth}
      \centering
      \infrule{ \Gamma\vdash\mathtt{fun}\; f(x:\tau) \;\{e'\}:\tau' \quad \Gamma\vdash e:\tau}{\Gamma\vdash f(e):\tau'}
    \end{minipage}
\end{equation*}
  \caption{Typing rules for \textit{Dice}~\cite{dicepaper}, a simple language for discrete probabilistic programming. Metavariable $f$ ranges over function names, $x$ over variable names, and $p$ over reals in the range $[0,1]$.}\label{fig:typing-rules}
\end{figure}

At the level of the purely Boolean fragments, Dice and our circuits are clearly equally expressive, though Dice uses \texttt{if-then-else} as primitive, rather than conjunction and negation. This is only a superficial difference, since both are well known to constitute universal sets of gates for Boolean functions. Instead of wires, Dice's syntax uses \emph{variables} to encode circuits inputs and \texttt{let}-expressions to encode circuit composition. Note that the latter are not needed to represent plain Boolean functions---the purely algebraic syntax of Boolean algebra suffices for this purpose~\cite{boole?}. However, once we add probabilistic effects, they are necessary, as we now explain.

Dice's first probabilistic primitive is $\Bernoulli{p}$, which represents a coin flip with probability $p\in[0,1]$ of landing on \texttt{true}, like our own $\Flip{p}$. Note that simply adding $\Bernoulli{p}$ to the syntax of Boolean algebra would give a syntax that is insufficient to reuse (or ignore) the outcome of a $\Bernoulli{p}$. 
Indeed there would be no way of binding its outcome to some variable for later use, contrary to our diagrammatic syntax, which can broadcast the output of $\Flip{p}$ to multiple sub-circuits using $\Bcomult$ (or discard it with $\Bcounit$). This is why \textit{Dice} introduces \texttt{let}-expressions into the syntax: to mimic the non-linear use of probabilistic outcomes. Concretely, $\Let{x}{e_0}{e_1}$ binds some term $e_0$ to the variable $x$ for its use in $e_1$. 

Dice's second probabilistic primitive is $\texttt{observe } a$, interpreted as evidence that $a$ is true; it has the effect of conditioning the posterior distribution the program represents on the $a$ being true (or, of setting to zero the probability of executions on which $a$ is not true). Our circuits have a primitive which plays a similar role: $\Bmult$ constrains its two inputs to have equal value. Thus, $\texttt{observe } a$ can be encoded by setting one of the inputs of $\Bmult$ to $\Flip{1}$ (true) and composing the remaining wire with the encoding of $a$ as a circuit.

Based on this discussion, we can devise a more formal translation of the symbolic PL into circuits. 
A (partial) program $\Gamma \vdash e\colon \tau$ will now be translated into a circuit $\diagsem{e}\from |\Gamma|\to |\tau|$ where $|\Gamma|$ is simply the length of the context $\Gamma$ and, similarly, $|\tau|$ is the length of the corresponding tuple of Booleans. In other words, a partial program with $m$ free variables and whose main expression has type $\Bool^n$ represents a circuit with $m$ input wires and $n$ output wires. The translation is then given inductively on the typing rules in Fig.~\ref{fig:pl-translation}. Note that it makes use of the multiplexer and $n$-ary gates defined earlier as syntactic sugar.

\begin{example}
  We provide an identical scenario to that presented in Example 1, this time using Dice's syntax.
  
  \begin{minipage}{0.7\linewidth}
    \begin{lstlisting}[mathescape=true, basicstyle=\ttfamily\small]
      let x = flip 0.1  in
      let y = if x then flip 0.2 else flip 0.3 in
      let z = if y then flip 0.4 else flip 0.5 in 
      z
    \end{lstlisting}
  \end{minipage}
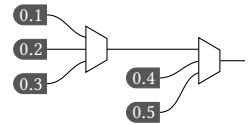
\begin{minipage}{0.3\linewidth}
\[
\InputIfFileExists{example-causal.tikz}{}{\input{./tikz/example-causal.tikz}}
\]
\end{minipage}
\end{example}

\begin{figure}
  \begin{equation*}
    
\InputIfFileExists{var-intro-diag.tikz}{}{\input{./tikz/var-intro-diag.tikz}}
\;=\;  
\InputIfFileExists{variable.tikz}{}{\input{./tikz/variable.tikz}}
\qquad \quad      
\InputIfFileExists{let-intro-diag.tikz}{}{\input{./tikz/let-intro-diag.tikz}}
\; =\; 
\InputIfFileExists{let.tikz}{}{\input{./tikz/let.tikz}}

  \end{equation*}
  \begin{equation*}
      
\InputIfFileExists{if-then-else-diag.tikz}{}{\input{./tikz/if-then-else-diag.tikz}}
 \; =\;
\InputIfFileExists{if-then-else-intro.tikz}{}{\input{./tikz/if-then-else-intro.tikz}}

  \end{equation*}
      \begin{equation*}
  
\InputIfFileExists{tuple-intro-diag.tikz}{}{\input{./tikz/tuple-intro-diag.tikz}}
\;=\;
\InputIfFileExists{tuple-intro.tikz}{}{\input{./tikz/tuple-intro.tikz}}

  \end{equation*}
  \begin{equation*}
      
\InputIfFileExists{fst-projection-diag.tikz}{}{\input{./tikz/fst-projection-diag.tikz}}
\;=\;
\InputIfFileExists{fst-projection.tikz}{}{\input{./tikz/fst-projection.tikz}}

\qquad\qquad
        
\InputIfFileExists{snd-projection-diag.tikz}{}{\input{./tikz/snd-projection-diag.tikz}}
\;=\;
\InputIfFileExists{snd-projection.tikz}{}{\input{./tikz/snd-projection.tikz}}

  \end{equation*} 
  \begin{equation*}
      
\InputIfFileExists{flip-intro-diag.tikz}{}{\input{./tikz/flip-intro-diag.tikz}}
\; =\;
\InputIfFileExists{del-flip-p.tikz}{}{\input{./tikz/del-flip-p.tikz}}
\qquad \quad 
\InputIfFileExists{observe-diag.tikz}{}{\input{./tikz/observe-diag.tikz}}
\;=\;
\InputIfFileExists{observe-x.tikz}{}{\input{./tikz/observe-x.tikz}}

  \end{equation*}
    \begin{equation*}
  
\InputIfFileExists{function-intro-diag.tikz}{}{\input{./tikz/function-intro-diag.tikz}}
\;=\;
\InputIfFileExists{function-intro.tikz}{}{\input{./tikz/function-intro.tikz}}

  \end{equation*}
\begin{equation*}
  
\InputIfFileExists{function-application-diag.tikz}{}{\input{./tikz/function-application-diag.tikz}}
\;=\;
\InputIfFileExists{function-application.tikz}{}{\input{./tikz/function-application.tikz}}

  \end{equation*}
  \caption{Diagrammatic translation of  \textit{Dice}. The mapping $\diagsem{\cdot}$ is defined inductively: on types by $\diagsem{\tau_1\times \tau_2} = \diagsem{\tau_1}+\diagsem{\tau_2}$ with base case $\diagsem{\Bool} = 1$, on contexts by $\diagsem{\Gamma,x:\tau} = \diagsem{\Gamma}+\diagsem{\tau}$ with base case $\diagsem{\varnothing} = 0$, and on terms by the rules above.}\label{fig:pl-translation}
\end{figure}

\begin{example}
We come back to the urn puzzle of the introduction. The corresponding circuit is
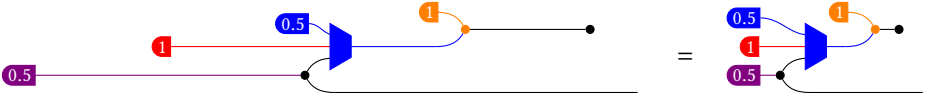
\[
\InputIfFileExists{urn-puzzle.tikz}{}{\input{./tikz/urn-puzzle.tikz}}
\]
Notice that the axioms of SMCs are already capturing non-trivial program transformations: for example, the first two \texttt{let}-expressions can be swapped without changing the resulting diagram:
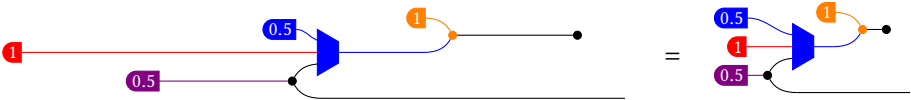
\[
\InputIfFileExists{urn-puzzle-swap.tikz}{}{\input{./tikz/urn-puzzle-swap.tikz}}
\]
This is one of the reasons we decide to work with a diagrammatic syntax---the topological transformations of the two-dimensional representation absorbs some of the burden of reasoning about program equivalence.
\end{example}
\begin{example}
Von Neumann's trick from Example~\ref{ex:von-neumann} can be translated into the following circuit:
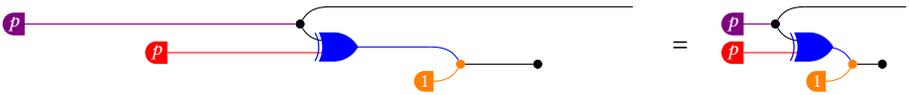
\[
\InputIfFileExists{von-neumann-xor.tikz}{}{\input{./tikz/von-neumann-xor.tikz}}
\]
where the blue gate is the usual exclusive \texttt{OR} (which can be encoded in the usual way using the other Boolean operations).
In fact, this turns out to be equivalent to the even simpler circuit below:
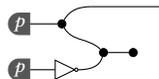
\[
\InputIfFileExists{von-neumann.tikz}{}{\input{./tikz/von-neumann.tikz}}
\]
Intuitively, this circuit flips two coins with the same parameter $p$ as in the previous circuit, negates the outcome of one of them, and conditions on the two resulting values to be equal and asks for the probability that the first flip gave true. 
\end{example}

\subsection{Semantics}
\label{sec:semantics}

In this paper,  `subdistribution' will always refer to a \emph{finitely-supported discrete probability  subdistribution}, that is, a map $\varphi\from X \to [0,1]$ where $X$ is finite and such that  $\sum_{x\in X} \varphi(x)\leq 1$. A \emph{distribution} is then a subdistribution $\varphi$ for which $\sum_{x\in X} \varphi(x)=1$. We will write (sub)distributions as  formal sums $\sum_x \varphi(x) \ket{x}$, omitting the elements of $X$ for which $\varphi(x)=0$, and call $\Dist{X}$ (resp. $\SubDist{X}$) the set of all distributions (resp. subdistributions) over some finite set $X$. 

\subsubsection{The symmetric monoidal category of (sub)stochastic maps.} \label{sec:smc-fstoch}
Our circuits do not map inputs to outputs in a deterministic way. Instead, for each input, they specify a subdistribution over their outputs. Thus, we will interpret them as \emph{substochastic maps}:
maps $f:X\to\SubDist Y$, which we write as $f\from X\distto Y$. 
Currying the $Y$ component, we can think of $f$ as a map $f(-|-)\from Y\times X\to [0,1]$ or, equivalently, as a matrix whose columns sum to at most $1$. We will switch freely between the two perspectives depending on which is more convenient to manipulate. As with distributions, a substochastic map is \emph{stochastic} when $\sum_y f(y|x) = 1$ for all $x\in X$ (or when its columns sum to $1$). Notice that any map $f\from X\to Y$ can be promoted 
to a stochastic map $X\distto Y$  
given by $x\mapsto \ket{f(x)}$ for all $x\in X$.

We can also interpret $f\from X\distto Y$ as specifying a \emph{conditional} subdistribution: $f(-|x)$ gives a subdistribution over $Y$ for each $x\in X$.  The composition $f\distcomp g$ of two substochastic maps $f:X\distto Y,g\from Y\distto Z$ 
is given by summing over the intermediate variable 
as follows:
\[(f\distcomp g)(z|x) = \sum_{y\in Y}g(z|y)f(y|x)\]
Notice that, if we think of $f$ and $g$ as substochastic matrices, this is the formula for their product in the usual sense. The identity $\id_X\from X\distto X$ is simply 
the Dirac delta $\delta_x\from x\mapsto \ket{x}$ or, in terms of matrices, the identity matrix. Substochastic maps with these operations define a category, which we call $\fSubStoch$~\cite{mossprobmonads}; since stochastic maps are stable under composition, they form a subcategory of $\fSubStoch$, which we call $\fStoch$~\cite{fritzmarkovcats}.

In addition, for two subdistributions $\varphi\in \SubDist X$ and $\rho\in\SubDist Y$, we can form the product subdistribution $\varphi
\otimes \rho \in\SubDist(X\times Y)$, given by $(\varphi\otimes \rho)(x,y) = \varphi(x)\cdot \rho(y)$. We will write $\ket{xy} = \ket{x}\otimes \ket{y}$ so that $\bigl(\sum_x \varphi(x)\ket{x}\bigr)\otimes \bigl(\sum_y \rho(y)\ket{y}\bigr) = \sum_{x,y} \varphi(x) \rho(y)\ket{xy}$. The same operation can be extended to conditional distributions, that is, to substochastic maps $f_1\from X_1\distto Y_1$ and $f_2\from X_2\distto Y_2$, giving $(f_1\otimes f_2)(x_1,x_2) = f_1(x_1)\otimes f_2(x_2)$. 
This makes $\fSubStoch$ into a \emph{symmetric monoidal} category, with the Cartesian product of sets as monoidal product, the singleton set  $1 =\{\bullet\}$ as unit, and the symmetry $\sigma_Y^X:X\times Y\distto Y\times X$ given by $\sigma_Y^X(x,y) = \ket{yx}$. Note that stochastic maps out of the unit $1$, \emph{e.g.} $1\distto X$, correspond precisely to distributions over $X$. Since the product of two distributions is still a distribution, $\fStoch$ is a symmetric monoidal subcategory of $\fSubStoch$.

\subsubsection{Copying and deleting.} \label{sec:Markov-category}
Notice that $\Dist(1) = 1$ and therefore, there is only one stochastic map $\epsilon_X \from X \distto 1$ for any set $X$, given by $\epsilon_X(x) = \ket{\bullet}$. Moreover, there is a canonical diagonal substochastic map inherited from the Cartesian product of sets: $\Delta_X\from X \distto X\times X$ given by $\Delta_X(x) = \ket{xx} :=\ket{x}\otimes\ket{x}$.
It is important to note that 
$f\distcomp \Delta_Y \neq \Delta_X \distcomp (f\times f)$ in general. We say that arbitrary substochastic maps cannot be copied. Intuitively, this makes sense: if $f$ represents some probabilistic process (such as flipping a coin), running $f$ once and reusing its outcome in two different places is different from running $f$ twice (flipping two coins). Those stochastic maps that do satisfy $f\distcomp \Delta_Y = \Delta_X \distcomp (f\times f)$ are precisely the \emph{deterministic} maps, that is, %
maps for which $f(x) = \ket{y}$ for a single $y\in Y$. 

Similarly, not every substochastic map $f\from X\distto Y$ satisfies  $f\distcomp\epsilon_Y = \epsilon_X$. Those that do are precisely the \emph{stochastic} maps, since $(f\distcomp\epsilon_Y)(x) = \sum_{y\in X}f(y|x) = 1 = \epsilon_X(x)$!
Such maps are often called \emph{discardable} or \emph{causal} in the literature~\cite{fritzmarkovcats,cho2019disintegration}, which is why we refer to conditioning-free circuits as `causal circuits'. After we define their interpretation, we will see that they are precisely those circuits whose semantics give a bona-fide stochastic map. 

Furthermore, using $\epsilon$, we can recover the usual notion of \emph{marginal} from probability theory. Indeed, post-composing a stochastic map $f\from 1 \distto X\times Y$ with $\epsilon_Y$ involves summing over all $y\in Y$ as follows: $f\distcomp (\id_{X}\times\epsilon_Y)(x) = \sum_{y\in Y} f(x,y)$. If we think of $f$ as a joint distribution over variables taking values in $X$ and $Y$, the stochastic map $f\distcomp (\id_{X}\times\epsilon_Y)$ thus corresponds precisely to marginalising over $Y$ to obtain a distribution over $X$ only. 

\begin{remark}[The categorical corner]\label{rmk:Markov-categories}
The mapping $\Dist$ can be extended to a functor $\Set \to\Set$, which maps $f:X\to Y$ to $\Dist{f}:\Dist X\to\Dist Y$ defined by $\Dist f(\varphi)(y)=\sum_{x\in f^{-1}(y)}\varphi(x)$. Moreover, $\Dist$ can be equipped with the structure of a monad~\cite{giry1982categorical}, with unit $\eta^{\Dist}_X:X\to\Dist X$  given by $\eta^{\Dist}(x)=\ket{x}$,
and multiplication $\mu^{\Dist}_X:\Dist\Dist X\to\Dist{X}$ defined as $\mu^{\Dist}(\Phi)(x)=\sum_{\Phi(\varphi) > 0}\Phi(\varphi)\cdot\varphi(x)$. This is more commonly known as the distribution monad, and the morphisms of its Kleisli category $\Kl{\Dist}$ are sometimes referred to as Markov kernels. The same structure maps also make the mapping $\SubDist$ into a monad. As we will see, our semantics for causal (or conditioning-free) circuits will land in $\fStoch$, the full subcategory of $\Kl{\Dist}$ spanned by finite sets. We will need a quotient of  $\Kl{\SubDist}$ to interpret conditioning.

The product of (sub)distributions defined above endows $\Dist$ and $\SubDist$ with the structure of \emph{commutative} monads~\cite{kock1970monads}, turning  $\Kl{\Dist}$ and $\Kl{\SubDist}$ into symmetric monoidal categories. 

Moreover, a monad that satisfies $\Dist(1) = 1$ is called \emph{affine}~\cite{kock1971bilinearity} and, as a result, $\Kl{\Dist}$ and $\fStoch$ are both \emph{Markov categories}~\cite{fritzmarkovcats} with copy maps $\Delta_X$ and discarding maps $\epsilon_X$ as defined above, for each set $X$. In particular, this means that $\epsilon_X$ is a natural transformation or, equivalently, that the unit of the monoidal structure is terminal. However, as we saw, $\Delta$ does not define a natural transformation, since
$f\distcomp \Delta_Y \neq \Delta_X \distcomp (f\times f)$ in general. 

Finally, note that $\SubDist$ is not an affine monad, and that $\Kl{\SubDist}$ and $\fSubStoch$ are not Markov categories, since $\epsilon_X$ does not define a natural transformation in this case. These still retain the structure of \emph{gs-monoidal categories} (also known as a CD-categories), that is, symmetric monoidal categories with a supply of cocommutative comonoid objects~\cite{corradini1999algebraic}.
\end{remark}

\subsubsection{Conditionals, disintegration, and Bayes' law.}\label{sec:disintegration}
Similarly, we can translate the notion of \emph{conditioning} in the language of $\fSubStoch$. We have already mentioned that a stochastic map $f \from X\distto Y$ can be thought of as a conditional distribution $f(-|x)$ for each $x\in X$. 
Thus, fundamental facts about conditionals have their counterpart in $\fSubStoch$. In particular, given $f\from 1 \distto X\times Y$, there exists $f_{|X}\from X\distto Y$ such that $f\distcomp (\id_{X}\times\epsilon_Y)\distcomp \Delta_X \distcomp (\id_{X} \times f_{|X})$. This is a direct translation of the \emph{disintegration} of a joint distribution $\Prob(x,y) = \Prob(x)\Prob(y|x)$ into the product of a marginal and a conditional distribution. As we have seen, in the language of $\fStoch$, $f\distcomp (\id_{X}\times\epsilon_Y)$ corresponds to the marginal $\Prob(x)$, and $f_{|X}$ to the conditional $\Prob(y|x)$.

It is well-known that such disintegrations are \emph{not} unique in general. Indeed, when the  marginal distribution $\Prob(x)$ is not fully-supported, that is, when $\Prob(x) = 0$ for some $x\in X$, then $\Prob(y|x)$ can be arbitrary. It is clear that any two disintegrations of the same distribution (for the same order of the variables) can have conditionals that differ only on a set of measure zero for the corresponding marginal~\cite[Section 6]{cho2019disintegration}. %
\begin{proposition}[Almost-sure uniqueness of disintegrations]
\label{prop:disintegrations-as-unique}
If $(g,f_{|X})$ and $(g',f'_{|X})$ are two disintegrations of the same stochastic map $f\from A\distto X\times Y$, then $g=g'=f\distcomp (\id_{X}\times\epsilon_Y)$ and $f_{|X}(a,x) = f'_{|X}(a,x)$ for all $a\in A$ and $x\in X$ such that $g(x|a)=g'(x|a)\neq 0$.
\end{proposition}

Finally, the ability to disintegrate any joint distribution is intimately linked to the notion of \emph{Bayesian inference}: assume we have a stochastic map $f\from X\distto Y$, with prior belief about the distribution of its inputs, in the form of a distribution $p\from 1\distto X$, and some observation $y\in Y$ of the output of $f$---what can we infer about the (unobserved) input of $f$? Bayes' law gives us a recipe to answer this question and compute a new distribution over inputs (called the posterior) which accounts for the evidence $y$. In fact, for a fixed prior $p$, Bayes' law defines a stochastic map $\Bayes{f}{p}\from Y\distto X$, telling us how to obtain the posterior distribution for \emph{any} observation $y\in Y$. We call this map the \emph{Bayes inverse} of $f$ relative to $p$. When $f$ is invertible, its Bayesian inverse (relative to any prior) coincides with its actual inverse. The process of Bayesian inversion can also be thought of in terms of disintegration. Note that the joint distribution $p\distcomp \Delta_X\distcomp (\id_X\times f)\from 1 \distto X\times Y$ is already disintegrated into the product of a marginal $p\from 1 \distto X$ and a conditional $f\from X\distto Y$. But we can also disintegrate it as the product of a marginal over $Y$ and a conditional $Y\distto X$, given respectively by $p\distcomp f$ and $\Bayes{f}{p}$, the Bayesian inverse of $f$. Thus, $p\distcomp \Delta_X\distcomp (\id_X\times f) = p\distcomp f\distcomp\Delta_Y\distcomp (\Bayes{f}{p}\times \id_Y)$, which is just a restatement of the well-known equality $\Prob(x|y)\Prob(y) = \Prob(y|x)\Prob(x)$.

\subsubsection{Interpreting causal circuits.}\label{sec:semantic-functor}
To define the semantics of causal circuits below we will use the standard Boolean operations of conjunction (written `$\land$') and negation (written `$\lnot$'). 
\begin{definition}[Semantics of causal circuits]
\label{def:semantics}
Let $\sem{\cdot}$ be the mapping defined inductively on $\CausCirc$ by 
\begin{equation*}
\sem{\Bcomult\,}(x) = \ket{xx} \quad \sem{\Bcounit\,}(x) = \ket{\bullet} \quad  \sem{\Andgate}(x_1,x_2) = \ket{x_1\land x_2}\quad \sem{\Notgate}(x) = \ket{\lnot x}
\end{equation*}
\begin{equation*}
\sem{\,\Flip{p}}(\bullet) = p\ket{1} + (1-p)\ket{0}
\end{equation*}
\begin{equation}
\sem{\idone}(x) = \ket{x} \qquad \sem{\sym}(x,y) = \ket{yx} 
\end{equation}
\begin{align*}
\sem{
\InputIfFileExists{seq-compose.tikz}{}{\input{./tikz/seq-compose.tikz}}
}(z|x) &=  \sum_{y\in Y}\sem{\smallprobcircuit{d}{m}{n}}(z|y)\cdot \sem{\smallprobcircuit{c}{\ell}{m}}(y|x)
\\
\sem{\,
\InputIfFileExists{par-compose.tikz}{}{\input{./tikz/par-compose.tikz}}
\,}(y_1,y_2|x_1,x_2) &= \sem{\smallprobcircuit{c_1}{m_1}{n_1}}(y_1|x_1)\cdot \sem{\smallprobcircuit{c_2}{m_2}{n_2}}(y_2|x_2)
\end{align*}
\end{definition}
Note that the first line of Definition~\ref{def:semantics} is the usual semantics of these circuit gates, given in terms of Boolean operations, lifted to $\fStoch$.
The second line defines the semantics of $\Flip{p}$ to be a Bernoulli distribution with parameter $p$. The next line forces $\sem{\cdot}$ to map plain wires and wire crossings to identities and symmetries in $\fStoch$. Finally, the last two lines guarantee the \emph{compositionality} of our interpretation: the semantics of two circuits composed in sequence is their composition in $\fStoch$ and the semantics of two circuits in parallel is their monoidal product in $\fStoch$. Thus, a circuit $c\from m\to n$, with $m$ input wires and $n$ output wires, is interpreted as a stochastic map $\sem{c}\from \Bool^m\distto \Bool^n$, \emph{i.e.}, as a map $\Bool^m\distto \Dist(\Bool^n)$. As a result, this interpretation mapping does define a \emph{symmetric monoidal functor}. 
\begin{proposition}\label{prop:semantics-functorial}
The mapping $\sem{\cdot}$ defines a symmetric monoidal functor $\CausCirc\to \fStoch$.
\end{proposition}

\subsubsection{Semantics of observations.}\label{sec:observation-semantics}
Famously, explicit conditioning complicates the semantics of probabilistic programs. In Dice syntax, $\mathtt{observe}\;x$ is intended to condition the distribution encoded by the overall program in which it occurs on the value of $x$ being true. 
Similarly, we want the generator $\Bmult$ of our syntax to constrain both its inputs to have the same value. Hence, its interpretation should return the Dirac distribution $\ket{x}$ whenever its two inputs are both equal to $x$, but what should its output subdistribution be when its two inputs disagree? Following standard practice~\cite{bernoulliprobpaper,dicepaper,SemanticsOfProbabilisticPrograms}, we will assign probability $0$ to such failed constraints or observations. 
This suggests extending the interpretation $\sem{\cdot}$ to all of $\ProbCirc$ with
\begin{equation}
\sem{\Bmult}(y|x_1,x_2) = \begin{cases} 
1 &  \text{ if  $y=x_1=x_2$,}
\\ 
0 & \text{ otherwise.}
\end{cases} 
\end{equation}
Note that this does not define a distribution, but an unnormalised \emph{sub}distribution.
As a result, conditioning forces us to leave the realm of (normalised) distributions.

However, it is always possible to rescale a non-zero subdistribution to obtain a distribution, as is standard practice for the semantics of many PPLs~\cite{bernoulliprobpaper,dicepaper,fierens2014,olmedo2018conditioning}. And indeed, for the purpose of inference, any two closed probabilistic programs that represent the same distribution, should be considered equal. Equivalently, we would like to identify any two circuits $0\to n$ whose semantics give proportional subdistributions. While there are other interpretations of observation/conditioning in probabilistic programming~\citep{olmedo2018conditioning}, this is the perspective we adopt here. Let us make this idea more precise, following the work of Staton and Stein~\cite{stein2023probabilistic}.

For any two $\phi,\rho\in\SubDist X$, we write $\phi \propto\rho$ if there exists some real number $\lambda>0$ such that $\varphi(x) = \lambda\cdot \rho(x)$. It is easy to see that this defines an equivalence relation on $\SubDist X$. Moreover, since any non-zero subdistribution can always be normalised, the equivalence classes of $\propto$ are in one-to-one correspondence with (normalised) distributions over $X$, plus the uniformly zero subdistribution $\bot_X$, \emph{i.e.}, $\SubDist X/\propto\; \cong \Dist X + \{\bot_X\}$. In this sense, conditioning just adds the possibility of failure, denoted here by $\bot_X$, to the standard distributional semantics given by $\sem{\cdot}$. %

In summary, our intended semantics should identify circuits of type $0\to n$ whose subdistributional semantics differ only by a constant (non-zero) multiplicative factor. To interpret programs with free variables/circuits with input wires, we need to generalise the same idea to substochastic maps.   Moreover, we need to do so in a \emph{compositional} way: this means that the semantics for circuits $0\to n$ should extend to a semantics for arbitrary circuits $m\to n$ in a way which is compatible with sequential and parallel composition.

It turns out that our choice of identifying subdistributions up to a scalar factor forces us  to identify all subdistributions over the singleton set $1=\{\bullet\}$, and therefore all circuits $0\to 0$. But subdistributions over $1$ are precisely the substochastic maps $1\distto 1$, and these are in one-to-one correspondence with the non-negative reals. This tells us exactly how to extend our equivalence relation $\propto$ over subdistributions to substochastic maps: consider the substochastic maps $c \from 1\distto 1$ and $f\from X\distto Y$; then $c$ is equivalent to $id_1$, the identity map $1\distto 1$, which implies that $c \times f$ should be equivalent to $\id_{1}\times f = f$. Since $c$ was arbitrary, we have that any two substochastic maps that differ by a global nonzero multiplicative factor should be equivalent. 
\begin{definition}
\label{def:finprojstoch}
Given two substochastic maps $f,g\from X\distto Y$, we write $f\propto g$ if there exists a real number $\lambda > 0$ such that $f(y|x)=\lambda\cdot g(y|x)$ for all $x\in X$ and $y\in Y$. We write $[f]$ for the equivalence class of $f\from X\distto Y$.
\end{definition}
Note that the number $\lambda$ has to be same scalar factor \emph{for all} inputs in $X$, as the paragraph preceding Definition~\ref{def:finprojstoch} makes clear. See also Remark~\ref{rmk:global-normalise} below for a longer discussion on why allowing different scalar factors for each input does not give a compositional interpretation. And indeed, composition and product of stochastic maps are congruences for $\propto$, so that stochastic maps up to $\propto$ form a symmetric monoidal category, which we call $\fProjStoch$, with much of the same structure as $\fStoch$~\cite[Theorem 6.4]{stein2023probabilistic}. $\fProjStoch$ will be our target semantics for $\ProbCirc$, the syntax of all circuits, including those that contain explicit conditioning via $\Bmult$ generators. 
\begin{definition}[Semantics]
\label{def:cond-semantics}
For each generator $g$ of $\CausCirc$, let $\condsem{g}$ be the equivalence class of $\sem{g}$.
In addition, let $\condsem{\Bmult}$ be the equivalence class of the following subdistribution:
\[
\sem{\Bmult}(y|x_1,x_2) = \begin{cases} 
1 &  \text{ if  $y=x_1=x_2$,}
\\ 
0 & \text{ otherwise.}
\end{cases} 
\]
\end{definition}
As before, the interpretation of the whole syntax is functorial.
\begin{proposition}\label{prop:semantics-functorial}
The mapping $\condsem{\cdot}$ defines a symmetric monoidal functor $\ProbCirc\to \fProjStoch$.
\end{proposition}

\begin{remark}[On auto-normalisation]
\label{rmk:global-normalise}
In the paper that introduced \texttt{Dice}, the authors recognise, as we do, that a semantics given purely in terms of distributions is insufficient to interpret programs containing explicit conditioning~\cite[3.2.2]{dicepaper}. Instead, the semantics of a \texttt{Dice} program $t$ is -- in essence -- a pair containing a representation of the subdistribution encoded by $t$, and a representation of the normalising constant to turn it into a distribution. The authors argue that both are necessary and that re-normalising automatically is not sound.  However, the form of `auto-normalisation' they consider (to eventually reject it) coincides with ours on closed terms, but identifies programs with free variables with potentially different normalising factors for different values of their inputs. To see the difference, let us consider the two programs they give as examples:
\begin{equation}
\texttt{fun } f(x : \Bool) : \Bool\; \{ \texttt{let } z = x\lor \texttt{flip } 1/2 \texttt{ in let } z = \texttt{observe } y \texttt{ in } y \}
\end{equation}
\begin{equation}
\texttt{fun } g(x : \Bool) : \Bool\; \{ \texttt{true} \}
\end{equation}
Had we defined the semantics of \texttt{Dice} functions as the (normalised) distribution encoded by the function for each value of its input variables separately, then these two functions would be interpreted as the same stochastic map, which sends any input $x$ to the distribution $\ket{1}$. This is not sound, because there are contexts that can differentiate $f$ and $g$, such as 
\[\condsem{\diagsem{\texttt{let } x = \texttt{flip } 0.1 \texttt{ in let obs} = f(x) \texttt{ in } x}}(1) = 0.1/0.55\]
\[\condsem{\diagsem{\texttt{let } x = \texttt{flip } 0.1 \texttt{ in let obs} = g(x) \texttt{ in } x}}(1) = 0.1\]
whose distributional semantics is different. This is because $f$ has a `retroactive' effect on the distribution of $x$, while $g$ does not.

In our semantics however, $f$ and $g$ define circuits whose interpretation as substochastic maps are not proportional. Indeed, 
following the translation given in Section~\ref{sec:syntax}, they are mapped to the following two circuits:
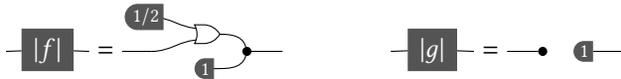
\begin{align*}
\probcircuit{\diagsem{f}}{}{}=
\InputIfFileExists{ex-circuit-conditional.tikz}{}{\input{./tikz/ex-circuit-conditional.tikz}}
\qquad \qquad\probcircuit{\diagsem{g}}{}{}=
\begin{tikzpicture}[circuit logic US]
	\begin{pgfonlayer}{nodelayer}
		\node [style=flip] (36) at (0.25, 0) {$1$};
		\node [style=none] (37) at (2, 0) {};
		\node [style=black] (38) at (-1.5, 0) {};
		\node [style=none] (39) at (-3, 0) {};
	\end{pgfonlayer}
	\begin{pgfonlayer}{edgelayer}
		\draw (36) to (37.center);
		\draw (38) to (39.center);
	\end{pgfonlayer}
\end{tikzpicture}
}

\end{align*}
For our chosen semantics, the first program is interpreted as the (equivalence class of the) substochastic map $\condsem{\diagsem{f}}\from \Bool\distto \Bool$, given by $\condsem{\diagsem{f}}(1) = \ket{1}$ and $\condsem{\diagsem{f}}(0) = \frac{1}{2}\ket{1}$. This is \emph{not} proportional to the subdistribution $\condsem{\diagsem{g}}(x) = \ket{1}$ for all $x$, since there is no scalar $\lambda$ which normalises $\condsem{\diagsem{f}}(x)$ to be equal to $\condsem{\diagsem{g}}(x)$ uniformly for \emph{both} $x=0$ and $x=1$.
\end{remark}

\begin{remark}[One-hot encoding]
A language that can model any distribution over some powers of the Booleans is also able to encode arbitrary distributions over finite sets, using what is commonly known as a \emph{one-hot encoding}. Semantically, given a distribution $\varphi$ over some finite set $X$, we can define a distribution $\Phi$ over $\Bool^X$ such that $\Phi(e_x)=\phi(x)$ whenever $e_x$ is the function that maps $x\in X$ to $1$ and the other elements of $X$ to $0$, and $\Phi(y)=0$ otherwise. 
On the syntax side, we encode a distribution over $X$ by a circuit with $|X|$ wires, whose semantics is a distribution over $\Bool^{|X|}$. Thus, each element of $X$ corresponds to a single (output) wire of the associated circuit, but different wires are not allowed to be set to \texttt{true} at the same time. One-hot encodings are ubiquitous in computer science. For example, the designers of the language Dice use it to define probabilistic programs over arbitrary finite sets of events~\cite{dicepaper}. Similarly, we can use our circuits and the associated equational theory to reason about the whole of $\fStoch$ (over arbitrary finite sets rather than just powers of the two-element set). 

The one-hot encoding can be extended straightforwardly to stochastic maps. It is therefore reasonable to hope that it defines a functor from $\fStoch$ to itself, but this is not the case. While this encoding preserves composition, it does not preserve identities or the monoidal product. One way to fix the issue would be to think of it as a functor from $\fStoch$ into its Karoubi envelope, whose objects are pairs $(X,e)$ of a set $X$ and an idempotent $e\from X\to X$. The categorical aspects of one-hot encoding are interesting in their own right, but they are not the main topic of this paper, so we will not pursue this conjecture formally here.
\end{remark}

\subsubsection{Props and symmetric monoidal theories}

Following standard practice, we can formalise circuits as morphisms of a product and permutation category or \emph{prop}. Formally, a prop is a strict SMC with the natural numbers as the set of objects and addition as the monoidal product. In a prop, all objects are thus monoidal products of a single generating object, \emph{viz.} 1.

Given a set of generating morphisms $\Sigma$, we can form the free prop $P_\Sigma$ as explained in Section~\ref{sec:term-to-diagrams}: by quotienting $\Sigma$-terms by the laws of SMC. For us here, $\ProbCirc$ is the free prop over the set of generators in~\eqref{eq:bool-syntax}-\eqref{eq:full-syntax} and $\CausCirc$ the free prop over those in~\eqref{eq:bool-syntax}-\eqref{eq:causal-syntax}.

Now that we have a formal syntax and semantics, our aim is to reason about semantic equivalence of circuits purely equationally. As we saw, since the syntax of circuits is a prop and its interpretation $\sem{\cdot}$ is a symmetric monoidal functor (Proposition~\ref{prop:semantics-functorial}) into another SMC, the laws of SMCs already capture some semantic equivalences between circuits, but they are not sufficient to derive all semantically valid equivalences. To this end, we need to add more.

Given a set $E$ consisting of equations between $\Sigma$-terms, we write $=_E$ for the smallest congruence with respect to the two compositions $\poi$ and $\otimes$ containing $E$. We call the elements of $E$ axioms and the pair $(\Sigma,E)$ a \emph{symmetric monoidal theory} (or more simply, \emph{theory}). Details on the existence and construction of free props on a given theory can be found in~\cite[Appendix B]{baez2018props} or~\cite{zanasi2015interacting}. We say that a theory is \emph{sound} if $c=_E d$ implies $\sem{c}=\sem{d}$ and \emph{complete} when the reverse implication also holds. A sound and complete theory is also called an \emph{axiomatisation}. When moreover, for every morphism $f$ of the target semantics there exists a morphism $c$ in $P_\Sigma$, the syntax, such that $\sem{c}=f$, we say that a sound and complete theory is a \emph{presentation} of the image of $\sem{\cdot}$.

In what follows we give a sound and complete theory for equivalence of causal circuits (Section~\ref{sec:axiomatisation-caus-circ}). Equivalently, this provides a presentation of $\fStoch_\Bool$, the SMC of stochastic maps between sets that are powers $\Bool$. For our second main result, we extend this theory to axiomatise the full language of probabilistic circuits, including explicit conditioning (Section~\ref{sec:observe}).

\section{An Axiomatisation of Causal Circuits}
\label{sec:axiomatisation-caus-circ}

The main contribution of our work is to provide a sound and complete axiomatisation of $\ProbCirc$ for the semantics given above. We begin by focusing on the causal part, $\CausCirc$, the axiomatisation of which is of independent interest, since it gives a presentation of a monoidal subcategory of $\fStoch$. 
The axiomatisation of circuits with conditioning will be delayed to Section~\ref{sec:observe}.

\subsection{Equational Theory}
\label{sec:equational-theory}

We will formulate our results using string diagrams, which simplify the burden of simultaneously using sequential and parallel composition, while at the same allowing us to build on top of existing axiomatisation results \cite{lafont2003}. We consider $\CausCirc$ terms quotiented by the set of equations in Fig.~\ref{fig:eqs} and the laws of symmetric monoidal categories (~\cref{eq:smc}). Note that any axioms which uses parameters $p,q,\dots$ implicitly quantifies over $[0,1]$, unless mentionned otherwise (as in, \emph{e.g.}, axiom E2).

\begin{figure}[htb]
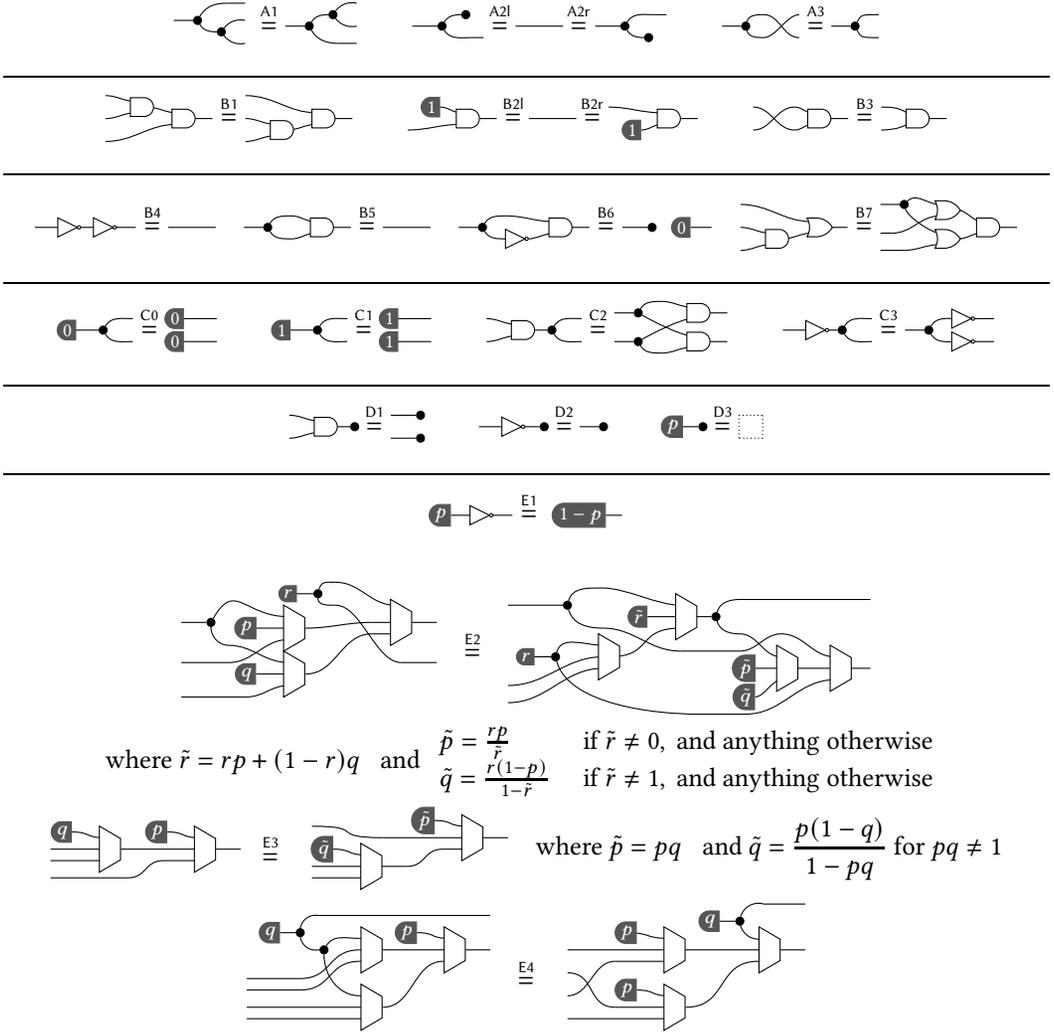

  \begin{subfigure}{\linewidth}
    \centering
    \begin{gather*}
      
\InputIfFileExists{copy-associative.tikz}{}{\input{./tikz/copy-associative.tikz}}
 \;\myeq{A1} 
\InputIfFileExists{copy-associative-1.tikz}{}{\input{./tikz/copy-associative-1.tikz}}
 \qquad 
\InputIfFileExists{copy-unital-left.tikz}{}{\input{./tikz/copy-unital-left.tikz}}
 \myeq{A2l} \idone\myeq{A2r} 
\InputIfFileExists{copy-unital-right.tikz}{}{\input{./tikz/copy-unital-right.tikz}}
\qquad  
\InputIfFileExists{copy-commutative.tikz}{}{\input{./tikz/copy-commutative.tikz}}
 \myeq{A3} \Bcomult
    \end{gather*}
  \end{subfigure}
  \noindent\rule{\linewidth}{0.4pt}
  \begin{subfigure}{1\linewidth}
    \centering
    \begin{gather*}
      
\InputIfFileExists{and-associative-left.tikz}{}{\input{./tikz/and-associative-left.tikz}}
\myeq{B1}
\InputIfFileExists{and-associative-right.tikz}{}{\input{./tikz/and-associative-right.tikz}}
\qquad 
\InputIfFileExists{and-unit-left.tikz}{}{\input{./tikz/and-unit-left.tikz}}
\myeq{B2l}\idone \myeq{B2r}
\InputIfFileExists{and-unit-right.tikz}{}{\input{./tikz/and-unit-right.tikz}}
\qquad 
\InputIfFileExists{swap-and.tikz}{}{\input{./tikz/swap-and.tikz}}
\myeq{B3}\Andgate
    \end{gather*}
  \end{subfigure}
  \noindent\rule{\linewidth}{0.4pt}
  \begin{subfigure}{\linewidth}
    \centering
    \begin{gather*}

\InputIfFileExists{not-not.tikz}{}{\input{./tikz/not-not.tikz}}
 \myeq{B4} \idone\quad 
\InputIfFileExists{copy-and.tikz}{}{\input{./tikz/copy-and.tikz}}
\myeq{B5} \idone \quad

\InputIfFileExists{copy-idxnot-and.tikz}{}{\input{./tikz/copy-idxnot-and.tikz}}
 \myeq{B6} 
\InputIfFileExists{del-false.tikz}{}{\input{./tikz/del-false.tikz}}
\quad 
\InputIfFileExists{and-or.tikz}{}{\input{./tikz/and-or.tikz}}
\myeq{B7} 
\InputIfFileExists{orx2-and.tikz}{}{\input{./tikz/orx2-and.tikz}}

    \end{gather*}
  \end{subfigure}
  \noindent\rule{\linewidth}{0.4pt}
  \begin{subfigure}{\linewidth}
    \centering
    \begin{gather*}

\InputIfFileExists{0-copy.tikz}{}{\input{./tikz/0-copy.tikz}}
\myeq{C0} 
\InputIfFileExists{0x0.tikz}{}{\input{./tikz/0x0.tikz}}
\qquad 
\InputIfFileExists{1-copy.tikz}{}{\input{./tikz/1-copy.tikz}}
\myeq{C1} 
\InputIfFileExists{1x1.tikz}{}{\input{./tikz/1x1.tikz}}
\qquad 
\InputIfFileExists{and-copy.tikz}{}{\input{./tikz/and-copy.tikz}}
\myeq{C2} 
\InputIfFileExists{copyxcopy-andxand.tikz}{}{\input{./tikz/copyxcopy-andxand.tikz}}
\qquad 
\InputIfFileExists{not-copy.tikz}{}{\input{./tikz/not-copy.tikz}}
\myeq{C3} 
\InputIfFileExists{copy-not.tikz}{}{\input{./tikz/copy-not.tikz}}

    \end{gather*}
  \end{subfigure}
    \noindent\rule{\linewidth}{0.4pt}
    \begin{subfigure}{\linewidth}
    \centering
    \begin{gather*}

\InputIfFileExists{and-del.tikz}{}{\input{./tikz/and-del.tikz}}
 \myeq{D1} 
\InputIfFileExists{delxdel.tikz}{}{\input{./tikz/delxdel.tikz}}
 \qquad 
\InputIfFileExists{not-del.tikz}{}{\input{./tikz/not-del.tikz}}
 \myeq{D2}\Bcounit \qquad 
\InputIfFileExists{flip-del.tikz}{}{\input{./tikz/flip-del.tikz}}
 \myeq{D3}
\InputIfFileExists{empty-diag.tikz}{}{\input{./tikz/empty-diag.tikz}}

    \end{gather*}
  \end{subfigure}
  \noindent\rule{\linewidth}{0.4pt}
  \begin{subfigure}{\linewidth}
    \centering
    \begin{gather*}
      
\InputIfFileExists{flip-not.tikz}{}{\input{./tikz/flip-not.tikz}}
\myeq{E1} \;
\InputIfFileExists{flip-1-p.tikz}{}{\input{./tikz/flip-1-p.tikz}}

      \\
      \\
      
\InputIfFileExists{ax-general-inverse.tikz}{}{\input{./tikz/ax-general-inverse.tikz}}
\\       \text{where }  \tilde{r}=rp+(1-r)q \;\;\text{ and }\begin{array}{ll}
          \tilde{p} = \frac{rp}{\tilde{r}}&\text{ if } \tilde{r}\neq 0, \text{ and anything otherwise} \\
         \tilde{q} = \frac{r(1-p)}{1-\tilde{r}} &\text{ if }  \tilde{r}\neq 1, \text{ and anything otherwise}
      \end{array} 
      \\
      
\InputIfFileExists{ax-bary.tikz}{}{\input{./tikz/ax-bary.tikz}}
\quad \text{where } \tilde{p}=pq \;\;\text{ and } \tilde{q}=\frac{p(1-q)}{1-pq}\text{ for }pq\neq 1
      \\
      
\InputIfFileExists{ax-swap.tikz}{}{\input{./tikz/ax-swap.tikz}}

    \end{gather*}
  \end{subfigure}
  \caption{Axioms of the causal circuits.}
  \label{fig:eqs}
\end{figure}

A few comments are in order, to clarify the meaning of the axioms in each block
\begin{itemize}
\item In the first block, the A-axioms define a \textit{cocommutative comonoid} consisting of a \textit{comultiplication} $\Bcomult:1\to 2$ and a \textit{counit} $\Bcounit:1\to 0$, together with equations guaranteeing that all different ways of copying a value are equal, that copying followed by discarding is the same as doing nothing, and that two copies are identical, so can be exchanged.
\item In the second block, the axioms (B1)-(B3) define a \textit{commutative monoid} consisting of two generators in our signature, a \textit{multiplication}, given by the Boolean \texttt{AND}-gate $\Andgate:2\to 1$, and a \textit{unit}, the Boolean constant \texttt{true}, represented by $\Flip{1}:0\to 1$, with their associated axioms: $\Andgate$ is associative, commutative and its unit is $\Flip{1}$.
\item In the third block, familiar axioms from Boolean algebra are given, stating the \emph{involution} of negation, the \emph{idempotence} of conjunction, as well as the usual \emph{complementation} axiom and the \emph{distributivity} of conjunction over disjunction.
\item In the third block, the axioms allow us to copy (C) arbitrary Boolean operations using $\Bcomult$. In other words, conjunction, negation, and the two Boolean constants \texttt{true} ($1$) and \texttt{false} ($1$) distribute over $\Bcomult$. Semantically, this charaterises them as \emph{deterministic} maps (or constants) without any probabilistic behaviour. 
\item In the fifth block, the D-axioms allow us to discard (D) the result of Boolean operations, using $\Bcounit$.  Note that axiom D3 applies more generally to $\Flip{p}$ for any $p\in[0,1]$ (not just $0$ and $1$) and encodes the normalisation of the distribution $\sem{\Flip{p}} = p\ket{1} + (1-p)\ket{0}$, \emph{i.e.} the fact that $p+(1-p) = 1$. Semantically, this characterises these as \emph{causal} operations. 
\item Note that the A-block and C-D blocks, combined with the algebraic theory of Boolean algebra in the B-block allow us to obtain a diagrammatic form of substitution for the purely Boolean fragment of the theory: it guarantees that any \emph{purely Boolean} expression can be copied (C) and discarded (D). In addition, they give us familiar algebraic structures that occur in similar diagrammatic calculi: $\Andgate,\Flip{1}$ with $\Bcomult,\Bcounit$ together form a (co)commutative \emph{bimonoid}.  
\item Finally, the last block contains all equations concerning the probabilistic behaviour of $\CausCirc$ terms,
  \par \textit{Axiom E1} is transparent: applying a negation after a coin flip with bias $p$ is the same as exchanging the probabilities of the corresponding distribution: $p\ket{0}+(1-p)\ket{1} = (1-p)\ket{\lnot 0}+p\ket{\lnot 1}$. This implies in particular that $\ket{0}$ is the negation of $\ket{1}$, which is a simple Boolean algebra identity.
  \par \textit{Axiom E2}  generalises a simpler equality, which is the diagrammatic counterpart of two ways of disintegrating a joint distribution of two variables, as $\Prob(y_0|y_1)\Prob(y_1)=\Prob(y_0)\Prob(y_1|y_0)$:
  \begin{equation}\label{eq:0-2-distribution}
  
\InputIfFileExists{nf-0-2-inv.tikz}{}{\input{./tikz/nf-0-2-inv.tikz}}
 \;= \;
\InputIfFileExists{nf-0-2.tikz}{}{\input{./tikz/nf-0-2.tikz}}

  \end{equation}
 E2 further conditions this identity on the value of some input variable (the top left wire). This more general form will be needed to iteratively disintegrate joint distributions over more than two variables, as we will see in Section~\ref{sec:m-n-completeness}. 
    By varying the values of $p_0,p_1,p_2$ in this last diagram, the semantics of the corresponding circuit ranges over all distributions on $\Bool^2$. 
    Since the \texttt{NOT}-gate is invertible (in fact involutive, as axiom B4 encodes), its Bayesian inverse is itself. As we will see in the proof of completeness, these axioms will allow us to compute the Bayesian inverse of any circuit, relative to any distribution.
  \par \textit{Axiom E3} can be understood as form of associativity for convex sums. Indeed, circuits of the form
    \[\ConvexSum{p}\]
    for some $p\in[0,1]$, are interpreted as the distribution $p\ket{x_0} + (1-p)\ket{x_1}$,  conditional on inputs $x_0$ and $x_1$. That is, circuits of this form simply take the convex sum of their two inputs! This is a binary operation, which is associative \emph{up to reweighing} as axiom E3 states. A similar axiom appears in previous work that contain convex sum (aka probabilistic choice) as a primitive~\cite{fritzpresentation,probgkat}. A difference with our work is that convex sum is a derived operation for us. This will allow us to derive laws commonly taken as axioms in related work, most notably the \emph{distributivity} of \texttt{if-then-else} over probabilistic choice~\cite[Section 5]{probgkat}. %
  \par \textit{Axiom E4} allows us to break redundant correlation between the guard of two \texttt{if-then-else} gates: notice that there is one fewer $\Bcomult$ node on the right than on the left. This ability will prove crucial in our proof of completeness, helping us to rewrite complicated circuits into simple trees of convex sums. 
\end{itemize}

\subsection{Soundness}
\label{sec:soundness-caus-circ}

Using our equational theory, everything syntactically provable remains true under our semantic interpretation. Proof of the below theorem is straightforward, as it is suffices to verify that each axiom in Fig.~\ref{fig:eqs} is a semantic equality.

\begin{theorem}[Soundness]
For all circuits $c,d\from m\to n$, if $c=d$ then $\sem{c}=\sem{d}$.
\end{theorem}
\begin{proof}
It suffices to verify that, for any axiom of the form $c=d$, we have $\sem{c} = \sem{d}$. We omit the verification of the Boolean algebra axioms which are standard, and of D3,E1 which are both immediate from the explanations above. We prove the soundness of axioms E2-E4 below.
\begin{description}
\item[(E2)]  We distinguish two cases: when the first wire is set to $1$, we have
\begin{align*}
\sem{
\InputIfFileExists{ax-general-inverse-left.tikz}{}{\input{./tikz/ax-general-inverse-left.tikz}}
}(1,x_1,x_2) &
\begin{array}{l}
 = rp\ket{11} + (1-r)q\ket{10} + r(1-p)\ket{01} + (1-r)(1-q)\ket{00} 
 \\
 = \tilde{r}\tilde{p}\ket{11} + \tilde{r}(1-\tilde{p})\ket{10} + (1-\tilde{r})\tilde{q}\ket{01} + (1-\tilde{r})(1-\tilde{q})\ket{00}
\end{array} 
\\
& = \sem{
\InputIfFileExists{ax-general-inverse-right.tikz}{}{\input{./tikz/ax-general-inverse-right.tikz}}
}(1, x_1,x_2)
\end{align*}
and when the first wire is set to $0$, we have
\begin{align*}
\sem{
\InputIfFileExists{ax-general-inverse-left.tikz}{}{\input{./tikz/ax-general-inverse-left.tikz}}
}(0,x_1,x_2) & =r\ket{x_11} + (1-r)\ket{x_20}  
\\ &= \sem{
\InputIfFileExists{ax-general-inverse-right.tikz}{}{\input{./tikz/ax-general-inverse-right.tikz}}
}(0, x_1,x_2)
\end{align*}
\item[(E3)] This axiom occurs in previous work~\cite{fritzpresentation,probgkat}, but in a different setting (and without proof of soundness), so we prove it here for reference.
\begin{align*}
\sem{
\InputIfFileExists{convex-sum-associate-right.tikz}{}{\input{./tikz/convex-sum-associate-right.tikz}}
}(x_0,x_1,x_2) &= pq\ket{x_0}+(1-pq)\frac{p(1-q)}{1-pq}\ket{x_1}+(1-pq)\left(1-\frac{p(1-q)}{1-pq}\right)\ket{x_2}
\\
& = pq\ket{x_0}+p(1-q)\ket{x_1}+(1-p)\ket{x_2}
\\
& = \sem{
\InputIfFileExists{convex-sum-associate-left.tikz}{}{\input{./tikz/convex-sum-associate-left.tikz}}
}(x_0,x_1,x_2)
\end{align*}
\item[(E4)] We have 
\begin{align*}
\sem{
\InputIfFileExists{ax-swap-left.tikz}{}{\input{./tikz/ax-swap-left.tikz}}
}(x_0,x_1,x_2,x_3) &=
\begin{array}{l}
pq\ket{1x_0}  + q(1-p)\ket{1x_2} \\
+ (1-q)p\ket{0x_1} + (1-q)(1-p)\ket{0x_2} 
\end{array} 
\\
& = \sem{
\InputIfFileExists{ax-swap-right.tikz}{}{\input{./tikz/ax-swap-right.tikz}}
}(x_0,x_1,x_2,x_3)
\end{align*}
\end{description}
\end{proof}

\subsection{Completeness}
\label{sec:completeness-caus-circ}

To prove completeness of our equational theory, we use a normal form argument, a common strategy in completeness proofs for diagrammatic calculi \cite{piedeleu2020,gu2023}. The idea is that normal forms provide a unique syntactic representative for each semantic object in the image of the interpretation functor $\sem{\cdot}$, thereby guaranteeing that, if two circuits denote the same distribution, they will be equal to the same normal form. The proof then relies on a normalisation procedure: an explicit algorithm which rewrites any given circuit into normal form, using only the proposed axioms. 

Our normalisation argument proceeds in two high-level steps: 1) first, we give a procedure to normalise $n\to 1$ causal circuits (Section~\ref{sec:n-1-compleness}); 2) then, we reduce the case of general $m\to n$ causal circuits to multiple applications of the procedure for the $n\to 1$ case (Section~\ref{sec:m-n-completeness}).
Let us explain each of these steps in a bit more details.
\begin{enumerate}
\item The normalisation procedure of $n\to1$ circuits can itself be broken down into two steps. \begin{itemize}
\item First, we rewrite a given $n\to 1$ circuit into \emph{pre-normal form}: a convex sum of Boolean circuits (Definition~\ref{def:pre-nf}). Recall that binary convex sums are a derivable operation in our syntax, given by an \texttt{if-then-else} gate with a single $\Flip{p}$ as its guard. Thus, the core of rewriting a given circuit into pre-normal form is \emph{i)} to move all occurrences of the probabilistic generators $\Flip{p}$ to the guard of some \texttt{if-then-else} gate, and  \emph{ii)} to eliminate correlations between different \texttt{if-then-else} gates which share the same $\Flip{p}$ (via $\Bcomult$ nodes). We will use \emph{Shannon expansion} to achieve \emph{i)}, but this tends to introduce further correlations between \texttt{if-then-else} gates, which we then have to remove with axiom E4. 
\item Then, given a circuit $c\from n\to 1$ in pre-normal form, we rewrite it into a normal form, which is just a direct encoding of the probability table of $\sem{c}\from \Bool^n\to \Bool$: it is a disjunction of all possible inputs and their associated Bernoulli distribution in the form of a single $\Flip{p}$ taken in conjunction (see Definition~\ref{def:nf-n-1}). Thus for any input $x\in\Bool^n$, we can read the probability of $\sem{c}(x)$ immediately from the normal form of $c$ and conclude that any two $n\to 1$ circuits that are mapped to the same distribution have the same normal form. Given $c\from n\to 1$ in pre-normal, \emph{i.e.}, given as a convex sum of $n \to 1$ Boolean circuits, Lemma~\ref{lem:derived-E5} (special case of axiom E2) will be our main tool to obtain the normal form, allowing us to add the contributions of the different Boolean components of the convex sum to the distribution of $\sem{c}$ on each input. 
\end{itemize}
\item For $m\to n$ causal circuits, the main idea is to reduce it to the $n\to 1$ case by reproducing the disintegration of the corresponding distribution. For example, we rewrite a circuit $c\from m\to 2$, into the composition of a circuit $m\to 1$ whose semantics is that of the marginal distribution on the variable corresponding to the first right wire, and a circuit $m+1\to 1$ whose semantics is that of the conditional distribution of the variable corresponding to the second wire, conditional on the first. Semantically, this amounts to disintegrating a joint distribution $\Prob(x,y)$ into the product $\Prob(x)\Prob(y|x)$, by conditioning on the first variable. This process can be iterated: starting from a circuit $m\to n$, we obtain $n$ circuits of type $m+k\to 1$, with $0\leq k\leq n-1$. We achieve this by repeatedly computing disintegrations by induction on the structure of circuits. 
Here, the main engine of the disintegration procedure is axiom E2, whose two sides give the two possible ways of disintegrating a joint distribution on two variables, as already explained in Section~\ref{sec:equational-theory}.
\end{enumerate}

\subsubsection{Preliminaries}
\label{sec:completeness-preliminaries}

First, we need to define syntactic sugar for $n$-ary versions of \texttt{AND}-gates to take the conjunction of $n$ inputs, and $n$-output copying nodes to copy one input to $n$ different locations:
\begin{align}
\label{eq:n-ary-broadcast}

\InputIfFileExists{broadcast-1-1xn.tikz}{}{\input{./tikz/broadcast-1-1xn.tikz}}
 := 
\InputIfFileExists{broadcast-1-1xn-def.tikz}{}{\input{./tikz/broadcast-1-1xn-def.tikz}}
\qquad 
\InputIfFileExists{broadcast-1-0.tikz}{}{\input{./tikz/broadcast-1-0.tikz}}
 := \Bcounit
\\

\InputIfFileExists{and-1xn-ary.tikz}{}{\input{./tikz/and-1xn-ary.tikz}}
 := 
\InputIfFileExists{and-1xn-ary-def.tikz}{}{\input{./tikz/and-1xn-ary-def.tikz}}
 \qquad 
\InputIfFileExists{and-0-ary.tikz}{}{\input{./tikz/and-0-ary.tikz}}
 := \Flip{1}
\end{align}
where the dotted lines represent $0$ wires. Once again, $n$-ary \texttt{OR}-gates can be defined in exactly the same way.

As a consequence of a more general correspondence between algebraic theories and symmetric monoidal theories where all generators can be copied and discarded, our theory is complete for Boolean circuits.
\begin{theorem}[Completeness for Boolean circuits]
\label{thm:boolean-completeness}
For any two purely Boolean circuits $c,d\from m\to n$, $c=d$ if and only if $\sem{c}=\sem{d}$.
\end{theorem}
\begin{proof}
This is an immediate consequence of a general theorem~\cite[Theorem 6.1]{bonchi2018deconstructing} establishing a one-to-one correspondence between algebraic theories and symmetric monoidal theories with a distinguished natural comonoid structure. Here, this structure is given by $\Bcomult$ and $\Bcounit$, which the A-block axioms make into a comonoid, as expected. Moreover, this comonoid is natural when all other generators distribute over $\Bcomult$ and $\Bcounit$; the axioms of the (B4)-(B7)-block guarantee precisely that. Finally, the axioms of the B-blocks and axiom C1 (for $p=1$) are a direct diagrammatic translation of the usual Boolean algebra axioms. Since the algebraic theory of Boolean algebras is well-known to be complete (for the standard semantics which coincides with ours) the general theorem applies and our theory is complete for Boolean circuits.
\end{proof}

\begin{remark}
Given Theorem~\ref{thm:boolean-completeness}, we can--and will--move freely between Boolean circuits and their semantics. This allows us to focus on the specifically probabilistic aspects of our calculus in the proof of completeness below.
\end{remark}

In what follows, we will make extensive use of \emph{Shannon expansion}, a notion we now recall. We write $f|_{x\leftarrow b}$ for the function resulting from assigning the Boolean value $b$ to its variable $x$, and say that $f|_{x\leftarrow b}$ is a \emph{restriction} of $f$.
When $f:\mathbb{B}^n\to\mathbb{B}$ is a Boolean function, given the only two possible restrictions for an argument $x_i$, $\texttt{true}$ and $\texttt{false}$, $f$ can be rewritten as
\[
f=\lnot x_i\cdot f|_{x_i\leftarrow 0}+x_i\cdot f|_{x_i\leftarrow 1}
\]
this identity is commonly referred to as the \textit{Shannon expansion} (or \textit{Shannon decomposition}) of $f$ w.r.t. $x_i$. 
\begin{lemma}[Shannon expansion]
  \label{lemma:shannon-expansion}
  For every Boolean circuit $b\from n\to 1$ with $n\geq 1$, we have
  \begin{center}
    
\InputIfFileExists{completeness/shannon-expansion-lemma.tikz}{}{\input{./tikz/completeness/shannon-expansion-lemma.tikz}}

  \end{center}
where $b_{\lnot x}$ and $b_x$ are the Boolean circuits that encode the $\Flip{0}$ and $\Flip{1}$ restrictions of $b$, respectively. 
\end{lemma}
\begin{proof}
This is a simple corollary of Boolean completeness (Theorem~\ref{thm:boolean-completeness}): since the identity holds in the semantics, the equality is derivable within our theory, for any Boolean circuit $b\from m\to 1$.
\end{proof}

\begin{lemma}[Boolean circuits can be copied]
\label{lem:copy-boolean}
Any boolean circuit $b\from m\to n$ satisfies the following identity:
\[
\InputIfFileExists{b-copy.tikz}{}{\input{./tikz/b-copy.tikz}}
\;=\;
\InputIfFileExists{copy-bxb.tikz}{}{\input{./tikz/copy-bxb.tikz}}
\]
We say that Boolean circuits can be copied.
\end{lemma}
\begin{proof}
Once again, this is a corollary of Theorem~\ref{thm:boolean-completeness}: since the equality holds plainly in the semantics, the equality is derivable within our theory, for any Boolean $b\from m\to n$.
\end{proof}
The following lemma tells us that $\CausCirc$ quotiented by our equational theory defines a Markov category~\cite{fritzmarkovcats} whose discarding maps are the family of $\Bcounitn{n}$ for all $n\in\N$ (which is the counit of the comonoid whose multiplication is the corresponding copying maps $\Bcomultn{n}$).
\begin{lemma}[Causal circuits can be discarded]
\label{lem:delete}
Any causal circuit $c\from m\to n$ satisfies the following identity:
\[
\InputIfFileExists{c-delete.tikz}{}{\input{./tikz/c-delete.tikz}}
\;=\;\Bcounitn{m}\]
We say that causal circuits are discardable.
\end{lemma}
\begin{proof}
We show this by induction on the structure of $c$. There is one base case for each generator of $\CausCirc$, each of which is handled by one axiom: A2l and A2r together for $\Bcomult$, B4 for $\Andgate$, B14 for $\Notgate$, B7 for $\Flip{p}$. Then, we have two inductive cases---one for sequential composition and one for parallel composition.
\begin{description}
\item[Sequential composition.] Assume circuits $c\from m\to n$ and $d\from n \to o$ can be discarded. Then
\[
\InputIfFileExists{c-d-delete.tikz}{}{\input{./tikz/c-d-delete.tikz}}
\;=\;
\InputIfFileExists{c-delete.tikz}{}{\input{./tikz/c-delete.tikz}}
\;=\;\Bcounitn{m}\]
\item[Parallel composition.] Assume circuits $c_1\from m_1\to n_1$ and $c_2\from m_2\to n_2$ can be discarded. Then
\[
\InputIfFileExists{c1xc2-delete.tikz}{}{\input{./tikz/c1xc2-delete.tikz}}
\;=\; 
\InputIfFileExists{c1-deletexdelete.tikz}{}{\input{./tikz/c1-deletexdelete.tikz}}
\;=\; 
\InputIfFileExists{deletexdelete.tikz}{}{\input{./tikz/deletexdelete.tikz}}
\]
\end{description}
\end{proof}
As a consequence of the last two lemmas (or of the completeness for Boolean circuits), Boolean circuits with multiple output wires can always be decomposed into multiple single-output circuits that share their inputs.
\begin{corollary}[Boolean circuits decompose]
\label{lem:boolean-decompose-multi-output}
For any Boolean circuit $b\from m\to n$, there exists $n$ Boolean circuits $b_1,\dots,b_n\from m \to 1$ such that
\[\boolcircuit{b}{m}{n} \;=\; 
\InputIfFileExists{b-decomposed.tikz}{}{\input{./tikz/b-decomposed.tikz}}
\]
\end{corollary}

Finally, we will need the following derived law, which is simply the single-output version of axiom E2. It will allow us to take the actual convex sum  (with weight $r$) of two $\Flip{p},\Flip{q}$, and will prove key in Section~\ref{sec:n-1-compleness} to reduce a tree of convex sums into a normal form from which the probability distribution associated to a given circuit can be read unambiguously.
\begin{lemma}\label{lem:derived-E5}
The following equality is derivable for all $r,p,q\in[0,1]$:
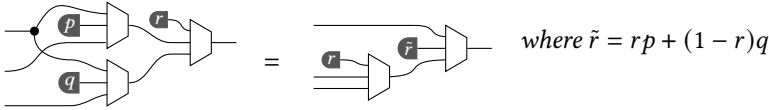
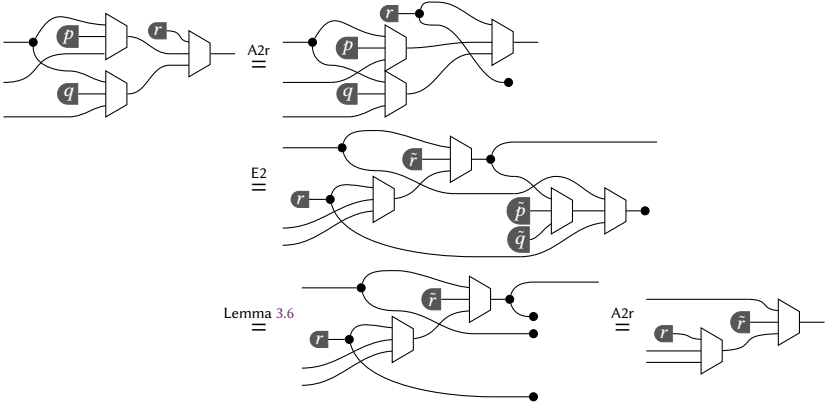
\[
\InputIfFileExists{ax-dist.tikz}{}{\input{./tikz/ax-dist.tikz}}
\quad \text{where } \tilde{r}=rp+(1-r)q \]
\end{lemma} 
\begin{proof}
\begin{align*}

\InputIfFileExists{ax-dist-left.tikz}{}{\input{./tikz/ax-dist-left.tikz}}
 \;& \myeq{A2r} \;
\InputIfFileExists{ax-dist-left-1.tikz}{}{\input{./tikz/ax-dist-left-1.tikz}}

\\
& \myeq{E2}\; 
\InputIfFileExists{ax-dist-left-2.tikz}{}{\input{./tikz/ax-dist-left-2.tikz}}

\\
& \myeq{Lemma~\ref{lem:delete}}\quad
\InputIfFileExists{ax-dist-left-3.tikz}{}{\input{./tikz/ax-dist-left-3.tikz}}
\; \myeq{A2r}\;
\InputIfFileExists{ax-dist-right.tikz}{}{\input{./tikz/ax-dist-right.tikz}}

\end{align*}
where we pick $\tilde{p},\tilde{q}$ that satisfy the side-conditions of axiom E2, depending on the value of $\tilde{r}$.
\end{proof}

\subsubsection{Pre-normal form}

Recall that the first step on the road to completeness is to show that every circuit can be rewritten into one in \textit{pre-normal form}. Intuitively, the pre-normal form that we adopt represents a diagram as a convex combination of Boolean expressions.\footnote{This can be seen as a generalisation of Birkhoff's theorem for doubly stochastic matrices~\cite{birkhoff} to the (singly-)stochastic case.} 
\begin{definition}[Pre-normal form]
\label{def:pre-nf}
A circuit $d\from n\to 1$ is in \emph{pre-normal form} if it is in the form defined inductively below:
    \begin{center}
      
\InputIfFileExists{pre-normal-form.tikz}{}{\input{./tikz/pre-normal-form.tikz}}

    \end{center}
where $b:n\to 1$ is a Boolean circuit and $d':n\to 1$ is itself in pre-normal form or is a Boolean circuit (base case).
\end{definition}
\begin{lemma}[Pre-normalisation]
  \label{lemma:pre-normalisation}
  Every $n\to 1$ causal circuit is equal to one in pre-normal form.
  \begin{proof}
    
We reason by structural induction on circuits of type $n\to 1$, that is, we show that for an arbitrary circuit  $d\from n\to 1$ in pre-normal form, composing with any of the generators in the monoidal signature of $\CausCirc$ results in a circuit that is equal to one in pre-normal form again.

We first show that the lemma holds for all generators. It is enough to show that $\Flip{p}$ can be put in this form, as every other generator is Boolean and thus already in pre-normal form. A simple Boolean algebraic identity gives us what we want:
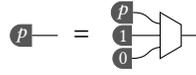
\begin{center}

\InputIfFileExists{completeness/n-to-1-flip.tikz}{}{\input{./tikz/completeness/n-to-1-flip.tikz}}

\end{center}

For the inductive case, consider an arbitrary circuit  $d\from n\to 1$ in pre-normal form, \emph{i.e.}, such that there exists a Boolean circuit $b\from n \to 1$ and some circuit $d'\from n\to 1$ in pre-normal form such that
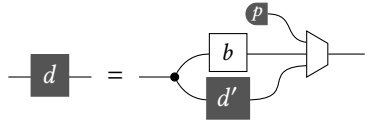
\begin{center}

\InputIfFileExists{completeness/n-to-1-pre.tikz}{}{\input{./tikz/completeness/n-to-1-pre.tikz}}

\end{center}

We now consider all possible ways to compose this circuit with a generator $g$ while preserving the circuit's arity. It is clear that any $n\to 1$ circuit can be obtained by doing any of the following two operations, for any generator $g$:
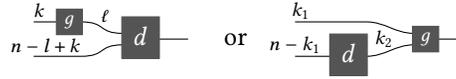
\begin{center}

\InputIfFileExists{completeness/n-to-1-g-comp-d.tikz}{}{\input{./tikz/completeness/n-to-1-g-comp-d.tikz}}

\end{center}
We want to show that for all of these cases, the resulting circuit is equal to one in pre-normal form whenever $d$ is.
We consider all possible cases below.

\begin{minipage}{.15\linewidth}
\begin{mdframed}
$g:=\Flip{q}$
\end{mdframed}
\end{minipage} In this case, $\Flip{q}$ can only be composed on the left. Without loss of generality, assume $d:n\to 1$ below has $k$ multiplexers. We have
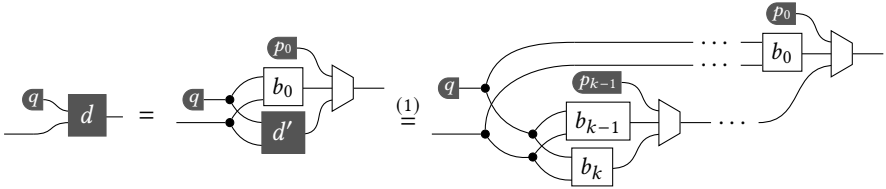
\begin{center}

\InputIfFileExists{completeness/n-to-1-flip-left01.tikz}{}{\input{./tikz/completeness/n-to-1-flip-left01.tikz}}

\end{center}
where the right-hand side of $(1)$ comes from fact that our circuits are finite and that the pre-normal form is defined inductively: thus, we can repeatedly expand $d'$ until we reach the base case of some Boolean circuit $b_k$ (for $k$ multiplexers).

First, we turn our attention to the two bottom-most Boolean circuits, $b_{k-1}$ and $b_{k}$, which we each Shannon expand w.r.t the first variable, corresponding to the wire into which $\Flip{q}$ is plugged:
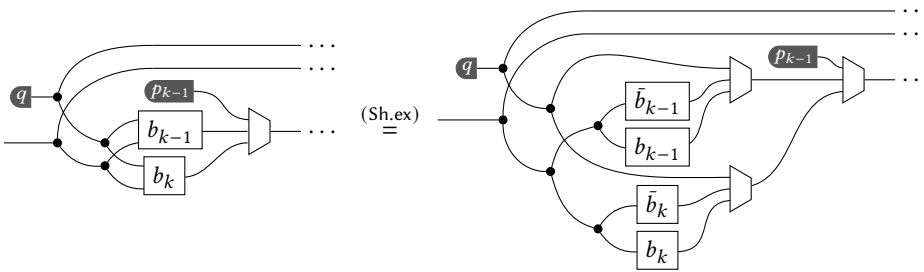
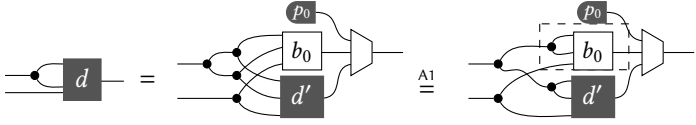
\[

\InputIfFileExists{completeness/n-to-1-flip-left02.tikz}{}{\input{./tikz/completeness/n-to-1-flip-left02.tikz}}

\]
We can now massage the last circuit slightly to apply axiom E4, and push $\Flip{q}$ forward:
\[

\InputIfFileExists{completeness/n-to-1-flip-left03.tikz}{}{\input{./tikz/completeness/n-to-1-flip-left03.tikz}}

\]
This procedure can be repeated for all remaining  $k-2$ components, until we reach the first multiplexer; then we use axiom E4 one last time, to obtain
\begin{center}

\InputIfFileExists{completeness/n-to-1-flip-left04.tikz}{}{\input{./tikz/completeness/n-to-1-flip-left04.tikz}}

\end{center}
At this point, intuitively, the resulting circuit is given as a tree of binary convex sums of Boolean circuits. However, this tree of convex sums is not associated in the correct way (our pre-normal form requires the convex sums to be fully associated to the right). Fortunately, we can correct this by repeatedly applying axiom E3, which gives us a form of associativity, up to reweighing the parameters of the various $\Flip{p}$ generators involved:
\begin{center}

\InputIfFileExists{completeness/n-to-1-flip-left05.tikz}{}{\input{./tikz/completeness/n-to-1-flip-left05.tikz}}

\end{center}
We repeat this process of re-associating the convex sums, until reaching the bottom-most multiplexer, at which point the circuit is in pre-normal form.

\begin{minipage}{.15\linewidth}
\begin{mdframed}
$g:=\Bcomult$
\end{mdframed}
\end{minipage}
\quad This case is straightforward, as pre-composing with a copier preserves the non-probabilistic structure of a circuit in pre-normal form. We can reason by induction on the number $k$ of multiplexer in the pre-normalised $d$. For $k=0$, the circuit $d$ is Boolean and therefore already in pre-normal form. Assume that we can pre-normalise 
\InputIfFileExists{copyxid-d.tikz}{}{\input{./tikz/copyxid-d.tikz}}
 for $d$ in pre-normal form with $k$  multiplexers. Now, let $d\from n\to 1$ be some circuit in pre-normal form with $k+1$ multiplexers. We have
\begin{center}

\InputIfFileExists{completeness/n-to-1-copy01.tikz}{}{\input{./tikz/completeness/n-to-1-copy01.tikz}}

\end{center}
The circuit in the dashed box above is Boolean, and we can apply the induction hypothesis to $d'$ (which has $k$ multiplexers) to obtain a circuit in pre-normal form.

\begin{minipage}{.165\linewidth}
\begin{mdframed}
$g:=\Andgate$
\end{mdframed}
\end{minipage} There are two possible cases for $\Andgate$: we can compose it on the left or on the right of $d$.
\begin{itemize}
\item First, on the left. We can reason once again by induction on the number $k$ of multiplexers in $d$. For the $k=0$, $d$ is Boolean and therefore already in pre-normal form. Assume that we can pre-normalise 
\InputIfFileExists{andxid-d.tikz}{}{\input{./tikz/andxid-d.tikz}}
 for $d$ in pre-normal form with $k$ multiplexers. Let $d\from n\to 1$ be some circuit in pre-normal form with $k+1$ multiplexers. We have
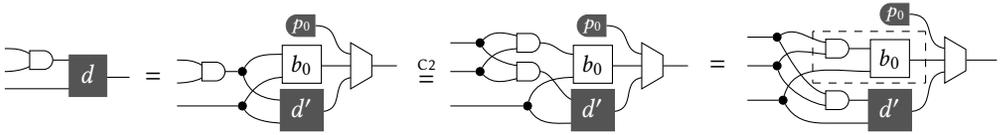
\begin{center}

\InputIfFileExists{completeness/n-to-1-and-left.tikz}{}{\input{./tikz/completeness/n-to-1-and-left.tikz}}

\end{center}
The circuit in the dashed box above is Boolean and we can apply the induction hypothesis to $d'$ to obtain a circuit in pre-normal form.

\item Now, on the right.  We reason again by induction on the number of multiplexers in $d$. For $k=0$, the circuit is Boolean so in pre-normal form. Assume that we can pre-normalise 
\InputIfFileExists{completeness/d-and.tikz}{}{\input{./tikz/completeness/d-and.tikz}}
 for all $d$ in pre-normal form with $k$ multiplexers. Let $d\from n\to 1$ be some circuit in pre-normal form with $k+1$ multiplexers. We will need the following easily verifiable Boolean algebra identity, which can be checked semantically (by Theorem~\ref{thm:boolean-completeness}) or algebraically, B-axioms:
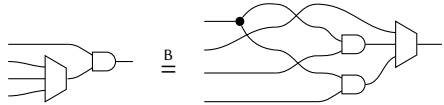
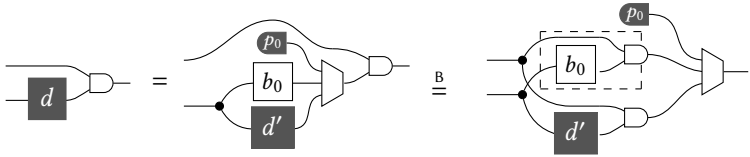
\[
\InputIfFileExists{completeness/if-and-identity.tikz}{}{\input{./tikz/completeness/if-and-identity.tikz}}
\]
Then, we have
\[
\InputIfFileExists{completeness/and-d-right.tikz}{}{\input{./tikz/completeness/and-d-right.tikz}}
\]
The circuit in the dashed box above is a Boolean circuit, and we can apply the induction hypothesis to $d$ to obtain a circuit in pre-normal form.

\end{itemize}

\begin{minipage}{.16\linewidth}
   \begin{mdframed}
      $g:=\Notgate$
   \end{mdframed}
\end{minipage}
Similar to the case of $\Andgate$ there are two cases for $\Notgate$: we show that pre and post-composing with $\Notgate$ results in a circuit in pre-normal form. 
\begin{itemize}
\item First, on the left. Once again, we reason by induction on the number of multiplexers in $d$. For the base case of $k=0$, the circuit is Boolean and thus in pre-normal form. Assume that we can pre-normalise $
\InputIfFileExists{completeness/not-d.tikz}{}{\input{./tikz/completeness/not-d.tikz}}
$ for all $d$ in pre-normal form with $k$ multiplexers. Let $d\from n\to 1$ be some circuit in pre-normal form with $k+1$ multiplexers. Then, we have

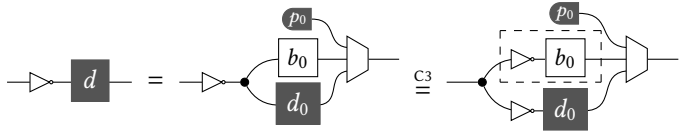
\begin{center}
   
\InputIfFileExists{completeness/not-d-left.tikz}{}{\input{./tikz/completeness/not-d-left.tikz}}

\end{center}
The circuit in the dashed box above is Boolean, and we can conclude the pre-normalisation by appealing to the induction hypothesis for $d'$. 

\item Now, on the right. We reason again by induction on the number of multiplexers in $d$. For $k=0$, the circuit is Boolean so in pre-normal form. Assume that we can pre-normalise 
\InputIfFileExists{completeness/d-not.tikz}{}{\input{./tikz/completeness/d-not.tikz}}
 for all $d$ in pre-normal form with $k$ multiplexers. Let $d\from n\to 1$ be some circuit in pre-normal form with $k+1$ multiplexers. Then, we have
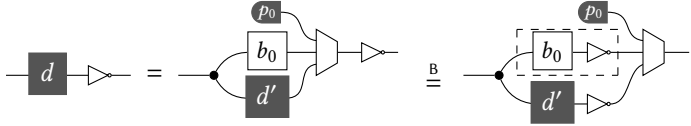
\begin{center}

\InputIfFileExists{completeness/not-d-right.tikz}{}{\input{./tikz/completeness/not-d-right.tikz}}

\end{center}
where the second equality is a simple Boolean algebra identity (it is clear that negating the output of \texttt{if-then-else} is the same as negating its two input branches). Finally, the circuit in the  dashed box above is Boolean, so we only have to apply the induction hypothesis to $d'$ to conclude the pre-normalisation proof.
\end{itemize}

  \end{proof}
\end{lemma}

\subsubsection{Completeness for single-output causal circuits}
\label{sec:n-1-compleness}

We have shown that any $n\to 1$ causal circuit can be expressed in pre-normal form. However, it is important to note that different circuits in pre-normal form may be mapped to the same stochastic map by the interpretation $\sem{\cdot}$. In this section we establish a method for showing that any two circuits in pre-normal form which are semantically equivalent, are equal. To do so, we define a suitable normal form for single-output circuits. Given some circuit $c\from n\to 1$, its normal form will be a simple syntactic encoding of the probability table of $\sem{c}$, \emph{i.e.}, of $\sem{c}(x)$ for all $x\in \Bool^n$. 

We will then show that any causal circuit of type $n\to 1$ is equal to one in normal form, by giving a procedure to rewrite it into normal form using the axioms of our theory. The key idea of this normalisation procedure is to start from the pre-normal form of a circuit, and to progressively aggregate all probabilities coming from the Boolean components of the convex sum that the pre-normal form represents. 
More specifically, for all possible bit-vector input $x\in\mathbb{B}^n$ to a circuit $\sem{c}\from\mathbb{B}^n\to\mathbb{B}$, we need to take a convex sum of all the weights associated coming from each of the Boolean components and multiplexer in the pre-normal form. This approach allows us to aggregate all the output weights corresponding to a particular input $x$, thereby reconstructing a circuit that explicitly encodes the probability table of $\sem{c}$.

For this purpose, we will need special circuits $\All{n}\from n\to 2^n$ designed to encode all possible $2^n$ bit vectors of length $n$, that is all possible inputs on $n$ wires.
\begin{definition}[All input circuits]
\label{def:all-inputs}
We define a family of circuits $\All{n}\from n\to 2^n$ by induction on $n$. Let $\All{0}:=\Flip{1}$ and
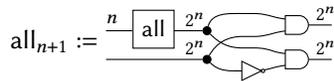
\[\All{n+1} := 
\InputIfFileExists{all-def-nx1.tikz}{}{\input{./tikz/all-def-nx1.tikz}}
\]
\end{definition}
We can use 
\begin{definition}[Normal Form for $n\to 1$ circuits]
  \label{def:nf-n-1}
  A circuit $n\to 1$ is in \emph{normal form} when it is of the form
  \[
\InputIfFileExists{n-1-nf.tikz}{}{\input{./tikz/n-1-nf.tikz}}
\]
where $c_0\from 2^n\to 1$ is a circuit defined recursively as follows: $c_{2^n} :=\Flip{0}$ and $c_{i} := 
\InputIfFileExists{n-1-if-cascade-def.tikz}{}{\input{./tikz/n-1-if-cascade-def.tikz}}
$ for some $p_i\in[0,1]$.
\end{definition}
\noindent One way to visualise the normal form more easily is as follows:
  \[
  
\InputIfFileExists{completeness/n-to-1-nf.tikz}{}{\input{./tikz/completeness/n-to-1-nf.tikz}}

  \]
  where each $a_i\from n\to 1$ corresponds to the circuit-encoding of the $i$-th bit vector of $\Bool^n$ in the order given by $\All{n}\from n \to 2^n$. 
\begin{lemma}[Completeness for $n\to 1$ causal circuits]
\label{lem:completeness-n-1}
  Every causal circuit of type $n\to 1$ is equal to one in normal form.
  \begin{proof}
By Lemma~\ref{lemma:pre-normalisation}, we can assume that all circuits we consider are already in pre-normal form and we will reason by induction on the number $k$ of multiplexers they have. For $k=0$, the circuit is Boolean, and therefore equal to a circuit in normal form by Boolean completeness (and in this case the weights $p$ in the $\Flip{p}$ generators of Definition~\ref{def:nf-n-1} are either $0$ or $1$).

Assume that all circuits of type $n\to 1$ in pre-normal form with up to $k$ multiplexers are equal to some circuit in normal form. Now, let $d\from n\to 1$ be a circuit in pre-normal form with $k+1$ multiplexers. By definition, there exists $b$ and $d'$ in pre-normal form (with necessarily $k$ multiplexers) such that
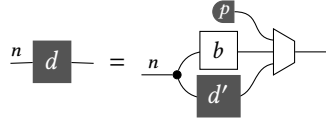
\[\probcircuit{d}{n}{}\;=\;
\InputIfFileExists{pre-normal-form.tikz}{}{\input{./tikz/pre-normal-form.tikz}}
\]
By Boolean completeness (Theorem~\ref{thm:boolean-completeness}), we can assume that $b$ is already in normal form. Thus, there exists $b_0\from 2^n\to 1$ such that $b = \All{n} ; b_0$, with $b_0$ satisfying the conditions of Definition~\ref{def:nf-n-1}. Moreover, by the induction hypothesis, there exists $c_0$ satisfying the same condition such that $d' = \All{n} ; c_0$. Therefore, we have
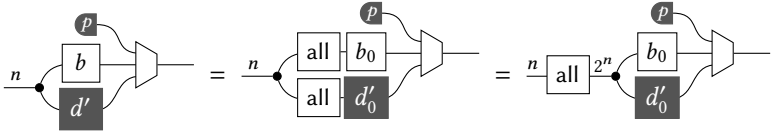
\[
\InputIfFileExists{pre-normal-form.tikz}{}{\input{./tikz/pre-normal-form.tikz}}
 \;=\; 
\InputIfFileExists{copy-all-pre-nf.tikz}{}{\input{./tikz/copy-all-pre-nf.tikz}}
 \;=\; 
\InputIfFileExists{all-pre-nf.tikz}{}{\input{./tikz/all-pre-nf.tikz}}
\]
where the last equality is an instance of Lemma~\ref{lem:copy-boolean}. Moreover, since both $b_0$ and $c_0$ are defined as in Definition ~\ref{def:nf-n-1}, there are circuits $b_1$, $d_1$ and weights $q_0$ and $r_0$ such that
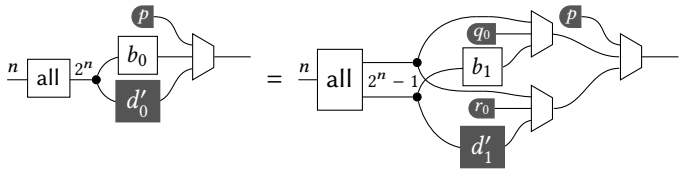
\[
\InputIfFileExists{all-pre-nf.tikz}{}{\input{./tikz/all-pre-nf.tikz}}
 \;=\; 
\InputIfFileExists{all-pre-nf-1.tikz}{}{\input{./tikz/all-pre-nf-1.tikz}}
\]
We can now apply Lemma~\ref{lem:derived-E5} (single-input version of axiom E2) to combine the weights:
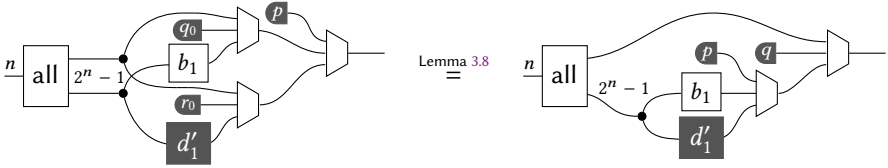
\[ 
\InputIfFileExists{all-pre-nf-1.tikz}{}{\input{./tikz/all-pre-nf-1.tikz}}
 \qquad\myeq{Lemma~\ref{lem:derived-E5}}\qquad 
\InputIfFileExists{all-pre-nf-2.tikz}{}{\input{./tikz/all-pre-nf-2.tikz}}
\]
where $q:= pq_0 + (1-p)r_0$. 

It is clear that we can iterate this procedure $2^n$ times (for all $b_i$ and $d'_i$ in the normal form of $b$ and $d'$) in order to obtain the normal form of $d$.%

  \end{proof}
\end{lemma}

\subsubsection{Completeness for arbitrary causal circuits}
\label{sec:m-n-completeness}

In this section we prove that our equational theory is complete for arbitrary $m\to n$ causal circuits, building on the results of the previous section on $n\to 1$ circuits. To do so, we will define a normal form for $m\to n$ causal circuits that uniquely characterises the distribution it represents, and prove that every such causal circuit is equal to one in normal form. As before, we say that a circuit equal to one in normal form is \emph{normalisable}.

\begin{definition}[Normal Form $m\to n$]
\label{def:nf-m-n}
A circuit $c\from m\to n$ is in \emph{normal form} when it is of the form
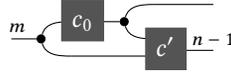
\[
\InputIfFileExists{nf-m-n-def.tikz}{}{\input{./tikz/nf-m-n-def.tikz}}
\]
where $c_0\from m\to 1$ in normal form in the sense of Definition~\ref{def:nf-n-1} and $c'\from m\to n-1$ is in normal form, with the following additional requirement: if $\sem{c_0}(b|x)$ has probability $0$ for $(b,x)\in\Bool\times\Bool^m$, then $\sem{c'}(b,x) =\ket{0^{n-1}}$.  
\end{definition}
The last condition of the above definition is here to deal with the non-uniqueness of disintegrations. As we saw in \S~\ref{sec:disintegration}, there are several possible disintegrations of a joint distribution $\Prob(x,y)$ into a marginal $\Prob(x)$ and a conditional $\Prob(y|x)$ if the marginal does not have full support---in this case, the value of $\Prob(y|x)$ for $x$ such that $\Prob(x)=0$ can be arbitrary. Because we want our normal forms to represent a given stochastic map uniquely, we need to fix a syntactic convention to handle such cases. Our chosen convention is to place all the probability mass on one arbitrary value, namely $0$, so that $\Prob(y|x) = \ket{0}$.
\begin{proposition}[Uniqueness of normal forms]
\label{prop:nf-unique}
Any two circuits $c,d\from m\to n$ in normal form and such that $\sem{c}=\sem{d}$ are equal.
\end{proposition}
\begin{proof}
We reason by induction on $n$. For $n = 0$, by Lemma~\ref{lem:delete}, there is only one circuit $m\to 0$ up to equality, namely $\Bcounitn{m}$ so the statement of the proposition holds. Assume that it holds for all circuits $m\to p$ for all $m$ and $p\leq n$. We want to show that any two circuits $c,d\from m\to n+1$ in normal form and such that $\sem{c}=\sem{d}$ are equal. By assumption, we have
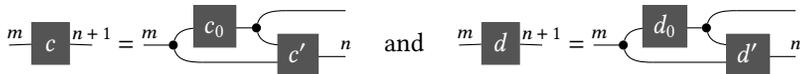
\[\probcircuit{c}{m}{\;\; n+1} = 
\InputIfFileExists{c-nf-m-nx1.tikz}{}{\input{./tikz/c-nf-m-nx1.tikz}}
\quad \text{ and }\quad \probcircuit{d}{m}{\; \; n+1} = 
\InputIfFileExists{d-nf-m-nx1.tikz}{}{\input{./tikz/d-nf-m-nx1.tikz}}
 \]
Let $f:=\sem{c}=\sem{d}$. Since $\sem{\cdot}$ is functorial, we have
\[f = \sem{c_0}\distcomp (\id_{\Bool}\times\epsilon_{\Bool^n})\distcomp \Delta_\Bool \distcomp (\id_{\Bool} \times \sem{c'}) = \sem{d_0}\distcomp (\id_{\Bool}\times\epsilon_{\Bool^n})\distcomp \Delta_\Bool \distcomp (\id_{\Bool} \times \sem{d'})\]
so that both normal forms provide a disintegration of the same stochastic map $f \from \Bool^{m} \distto \Bool^{n+1}$. Thus, by the almost-sure uniqueness of disintegrations (Proposition~\ref{prop:disintegrations-as-unique}), we have $\sem{c_0}=\sem{d_0}=: f_0$ and $\sem{c'}(b,x)=\sem{d'}(b,x)$ for all $(x,b)\in\Bool^{m}\times\Bool$ such that $f_0(b) \neq 0$. 

When $f_0(b) = 0$, our choice of normal form (see the final condition in Definition~\ref{def:nf-m-n}) ensure that $\sem{c'}(b,x)= \sem{d'}(b,x) = \ket{0^n}$. Thus, $\sem{c'}=\sem{d'}$ and, since $c',d'\from m\to n$ are both in normal form, the induction hypothesis allows us to conclude that $c'=d'$, as we wanted.
\end{proof}

The following theorem implies that every causal circuit is equal to one in normal form.
\begin{theorem}[Completeness]
\label{thm:completeness-m-to-n}
Every causal circuit of type $m\to n$ is equal to one in normal form.
\end{theorem}
\begin{proof}
  We show that, for any $m$ and $n$, every causal circuit $m\to n$ is equal to one in normal form. We first reason by induction on $n$, the circuit's number of output wires. The base case $n=0$ is immediate, as there is only one such causal circuit $m\to 0$ up to equivalence, by Lemma~\ref{lem:delete}, namely $\Bcounitn{m}$, which is in normal form.

We now assume we can normalise all causal circuits $m\to n$ for all $m$; we want to show that this holds for all causal circuits $m\to n+1$. For this, we reason by induction again, this time on the number of generators $\gamma$ a given causal circuits $m\to n+1$ has. For the base cases, note that every circuit consisting of $g\in\Sigma_\CausCirc\setminus\{\Bmult\}$, a generator in the monoidal signature of $\CausCirc$, is already in normal form. Now, assume all causal circuits $m\to n+1$ with $\gamma$ generators is normalisable. consider a causal circuit $c:m\to n+1$ made of $\gamma+1$ generators. Since it has $\gamma+1$ generators, we can pull the leftmost generator, $g$, from $c$, thus obtaining

  \begin{equation}\label{eq:g-d}
  \probcircuit{c}{m}{\;\;n+1}\; =\; 
\InputIfFileExists{completeness-m-to-n/g-d.tikz}{}{\input{./tikz/completeness-m-to-n/g-d.tikz}}

  \end{equation}

  where by hypothesis $\probcircuit{d}{m}{\;\;\;n+1}$ has $\gamma$ generators and is thus normalisable. As a result, there exists $d_0\from m\to 1$ and $d'\from m+1\to n$, such that
  \[\probcircuit{d}{m}{\; \; n+1} = 
\InputIfFileExists{d-nf-m-nx1.tikz}{}{\input{./tikz/d-nf-m-nx1.tikz}}
\]

  We now want to show that we can rewrite the circuit on the rhs of \eqref{eq:g-d} into a circuit in normal form, for all possible generators $g$---we consider all cases below.%

\noindent
  \begin{minipage}{.13\linewidth}
    \begin{mdframed}
      $g:=\Bcounit$
    \end{mdframed}
  \end{minipage}
  \[
\InputIfFileExists{completeness-m-to-n/bcounit-case.tikz}{}{\input{./tikz/completeness-m-to-n/bcounit-case.tikz}}
\]
where $d_0'\from m\to 1$ is some circuit in normal form, which we can get by completeness for $m\to 1$ circuits (Lemma~\ref{lem:completeness-n-1}). Moreover, by the induction hypothesis, the circuit in the dashed box on the right is normalisable since it is of type $m+1\to n$, giving us the normal form we wanted.
  
  \noindent
  \begin{minipage}{.16\linewidth}
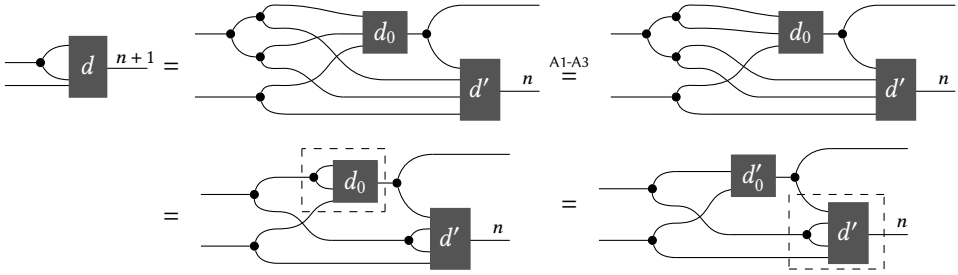
    \begin{mdframed}
      $g:=\Bcomult$
    \end{mdframed}
  \end{minipage}
  \[
\InputIfFileExists{completeness-m-to-n/bcomult-case.tikz}{}{\input{./tikz/completeness-m-to-n/bcomult-case.tikz}}
\]
 where once again, $d_0'\from m\to 1$ is some circuit in normal form, which we obtain by completeness for $m\to 1$ circuits (Lemma~\ref{lem:completeness-n-1}). Moreover, by the induction hypothesis, the circuit in the dashed box on the bottom right is normalisable since it is of type $m+1\to n$, giving us the normal form we wanted.

  \noindent
  \begin{minipage}{.165\linewidth}
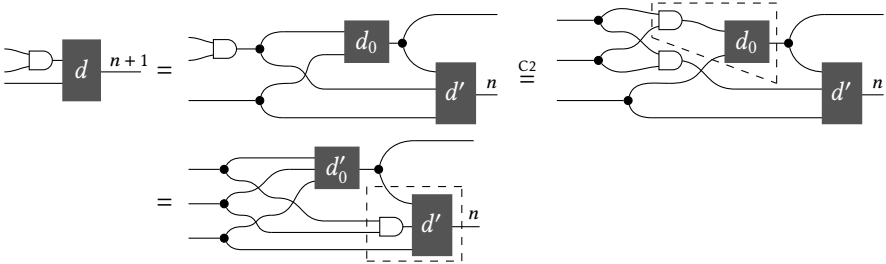
    \begin{mdframed}
      $g:=\Andgate$
    \end{mdframed}
  \end{minipage}
  \[
\InputIfFileExists{completeness-m-to-n/and-case.tikz}{}{\input{./tikz/completeness-m-to-n/and-case.tikz}}
\]
  where once again, $d_0'\from m\to 1$ is some circuit in normal form, which we obtain by completeness for $m\to 1$ circuits (Lemma~\ref{lem:completeness-n-1}). Moreover, by the induction hypothesis, the circuit in the dashed box in the last diagram is normalisable since it is of type $m+1\to n$, giving us the normal form we wanted.
  
  \noindent
  \begin{minipage}{.16\linewidth}
    \begin{mdframed}
      $g:=\Notgate$
    \end{mdframed}
  \end{minipage}
  \[
\InputIfFileExists{completeness-m-to-n/not-case.tikz}{}{\input{./tikz/completeness-m-to-n/not-case.tikz}}
\]
   where $d_0'\from m\to 1$ is some circuit in normal form, which we obtain by completeness for $m\to 1$ circuits (Lemma~\ref{lem:completeness-n-1}). Moreover, by the induction hypothesis, the circuit in the dashed box in the last diagram is normalisable since it is of type $m+1\to n$, giving us the normal form we wanted. %
  
  \noindent
  \begin{minipage}{.125\linewidth}
    \begin{mdframed}
      $g:=\sym$
    \end{mdframed}
  \end{minipage}
  \[
\InputIfFileExists{completeness-m-to-n/sym-case.tikz}{}{\input{./tikz/completeness-m-to-n/sym-case.tikz}}
\]
   where $d_0'\from m\to 1$ is some circuit in normal form, which we obtain by completeness for $m\to 1$ circuits (Lemma~\ref{lem:completeness-n-1}). Moreover, by the induction hypothesis, the circuit in the dashed box in the last diagram is normalisable since it is of type $m+1\to n$, giving us the normal form we wanted.%

  \noindent
  \begin{minipage}{.141\linewidth}
    \begin{mdframed}
      $g:=\Flip{p}$
    \end{mdframed}
  \end{minipage} \; This last case is the most complicated. Semantically, composing with $\Flip{p}$ amounts to taking the convex sum of the corresponding distributions, conditional on each possible values of the input wires:
  \[\sem{
\InputIfFileExists{completeness-m-to-n/flip-p-d.tikz}{}{\input{./tikz/completeness-m-to-n/flip-p-d.tikz}}
} = p\sem{
\InputIfFileExists{completeness-m-to-n/d-1xm-1xn.tikz}{}{\input{./tikz/completeness-m-to-n/d-1xm-1xn.tikz}}
} + (1-p)\sem{
\InputIfFileExists{completeness-m-to-n/d-1xm-1xn.tikz}{}{\input{./tikz/completeness-m-to-n/d-1xm-1xn.tikz}}
}\]
   Axiom E2 allows us to take this convex sum syntactically, and to repeat it for every possible value of the input wires (corresponding to a single multiplexer in the normal form of a given $m\to 1$ circuit), as we will see below.
  
  As before, since $d$ is normalisable, we have that
  \[
\InputIfFileExists{completeness-m-to-n/flip-case.tikz}{}{\input{./tikz/completeness-m-to-n/flip-case.tikz}}
\]

 Then, we can Shannon expand $d_0$ on the first input wire (the one connected to $\Flip{p}$), giving $d_0', d_1'$, such that
  \[
\InputIfFileExists{completeness-m-to-n/flip-case-01.tikz}{}{\input{./tikz/completeness-m-to-n/flip-case-01.tikz}}
\]

By completeness for single-output causal circuits (Lemma~\ref{lem:completeness-n-1}), we can assume that $d_0'$ and $d_1'$ are in normal form. Thus, there exists circuits $c_0,c_1$, $c'_0,c'_1$, and weights $q_0,q_1$, such that 
 \[
\InputIfFileExists{completeness-m-to-n/flip-case-02.tikz}{}{\input{./tikz/completeness-m-to-n/flip-case-02.tikz}}
\]
  We can now apply axiom E2:
  \[
\InputIfFileExists{completeness-m-to-n/flip-case-03.tikz}{}{\input{./tikz/completeness-m-to-n/flip-case-03.tikz}}
\]

where the weights $p',q'_0,q'_1$ are given as in the side condition of axiom E2 (see Section~\ref{sec:equational-theory}). We can repeat the same process for the second output wire of $\All{m}$; since $d'_0$ and $d'_1$ above were in normal form, we have circuits $c''_0,c''_1$ and weights $r_0,r_1$ such that
  \[
\InputIfFileExists{completeness-m-to-n/flip-case-04.tikz}{}{\input{./tikz/completeness-m-to-n/flip-case-04.tikz}}
\]
We can apply E2 again to the dashed box of rhs above, to obtain the following circuit:
  \begin{equation}
        
\InputIfFileExists{completeness-m-to-n/flip-case-05.tikz}{}{\input{./tikz/completeness-m-to-n/flip-case-05.tikz}}
\label{eq:flip-case-after-e2}
  \end{equation}
At this point, to obtain a circuit in normal form, we need to replace the additional $\Bcomult$ highlighted in red below with an identity wire:
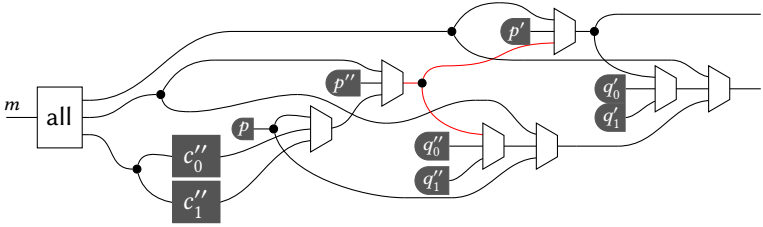
\[ 
\InputIfFileExists{completeness-m-to-n/flip-case-05-red.tikz}{}{\input{./tikz/completeness-m-to-n/flip-case-05-red.tikz}}
\]
As it turns out, this is possible using Boolean algebra axioms only. Indeed the two circuits below are equal---this can be checked semantically, followed by an appeal to Boolean completeness again:
  \[
\InputIfFileExists{completeness-m-to-n/level-2-boolean-eq-1.tikz}{}{\input{./tikz/completeness-m-to-n/level-2-boolean-eq-1.tikz}}
\]
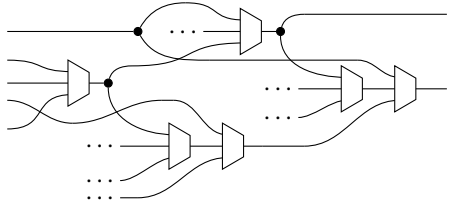
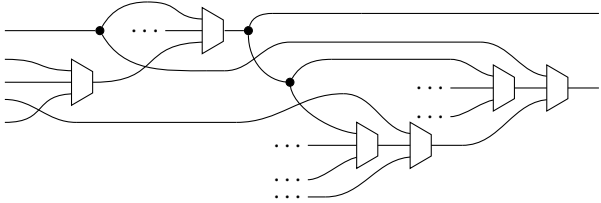
  \[
\InputIfFileExists{completeness-m-to-n/level-2-boolean-eq-2.tikz}{}{\input{./tikz/completeness-m-to-n/level-2-boolean-eq-2.tikz}}
\]
We thus have that~\eqref{eq:flip-case-after-e2} is equal to:
  \[
\InputIfFileExists{completeness-m-to-n/flip-case-06.tikz}{}{\input{./tikz/completeness-m-to-n/flip-case-06.tikz}}
\]
  This process can be repeated for each for each output wire of $\All{m}$: we can iteratively apply axiom E2 and remove the intermediate wires using Boolean algebra axioms, as explained above, until we reach the last such wire. Then,
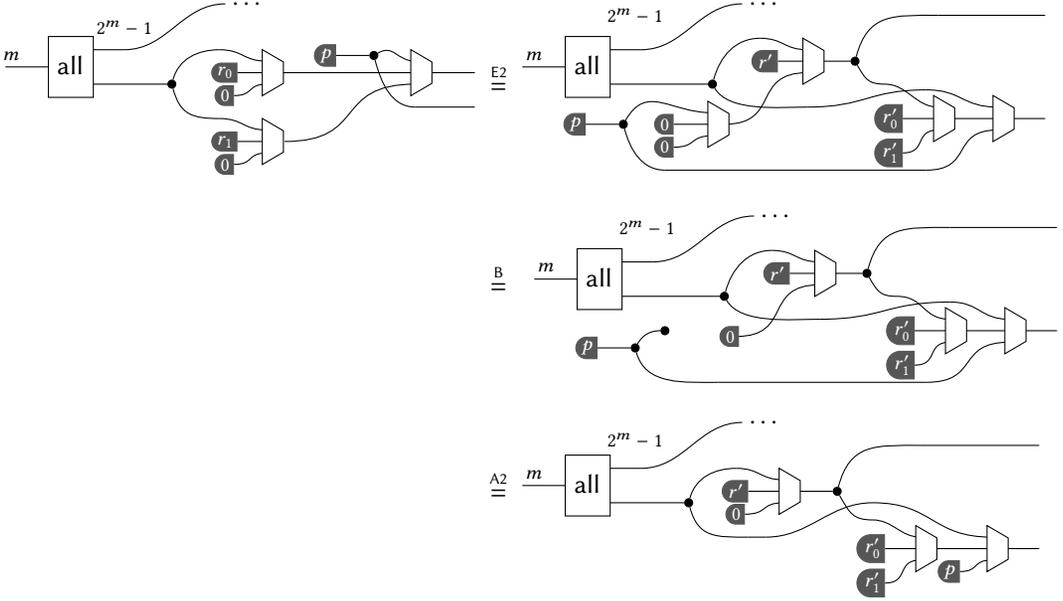
\[
\InputIfFileExists{completeness-m-to-n/flip-case-07.tikz}{}{\input{./tikz/completeness-m-to-n/flip-case-07.tikz}}
\]
Then, as we did above, we can move the rightmost $\Bcomult$ up to the output wire of $d_0$ using only Boolean algebra.
 The result of this procedure are circuits $d_0''\from m\to 1$ and $e\from 1\to 1$ such that:
  \[
\InputIfFileExists{completeness-m-to-n/flip-case-08.tikz}{}{\input{./tikz/completeness-m-to-n/flip-case-08.tikz}}
\] (Here $e$ is the cascade of if-then-else with as many guards as $\mathsf{all}:m\to 2^m$ has output wires, produced by the procedure above.)
As before, since $d_0''$ has a single output, it is normalisable, by Lemma~\ref{lem:completeness-n-1}. Finally, the circuit in the dashed box above has $n$ output wires and is therefore normalisable, by the induction hypothesis, giving us the required normal form and concluding the proof.
\end{proof}

By characterising the image of our semantic interpretation, $\sem{\cdot}:\CausCirc\to\fStoch$, we can see that our theory \emph{presents} $\fStoch_\Bool$, the full monoidal subcategory of $\fStoch$ on objects that are powers of the two-element set, \emph{i.e.}, $\mathbb{B}^n$ for $n\in\mathbb{N}$.

\begin{proposition}\label{prop:image-fstoch2}
  For every morphism $f$ of $\fStoch_\Bool$, there exists a causal circuit $c$ such that $\sem{c}=f$.
  The image of $\CausCirc$ under $\sem{\cdot}$ corresponds with $\fStoch_\Bool$.
  \begin{proof}
    First note that, by Definition~\ref{def:semantics}, we have that every causal circuit $c:m\to n$ corresponds to a stochastic map $\mathbb{B}^m\distto\mathcal{D}(\mathbb{B}^n)$ in $\fStoch$, associating to each $j\in\mathbb{B}^m$ a probability distribution on $\mathbb{B}^n$, and thus mapping $c$ to a stochastic matrix with dimensions $2^n\cong\mathbb{B}^n\times\mathbb{B}^m\cong 2^m$, where $m,n\in\mathbb{N}$. Moreover, we can represent \emph{any} stochastic map $f\from\Bool^m\distto\Bool^n$ as a causal circuit $m\to n$ for all $m,n\in\N$. Indeed, we can disintegrate $f$ in the order of its $n$ output variables, to obtain $n$ stochastic maps $f_1,\dots,f_n$. As we saw, the normal form for $k\to 1$ circuits (Definition~\ref{def:nf-n-1}) is simply an encoding of the probability table of a stochastic map with a single Boolean output variable. Thus, we can use it to define causal circuits $c_1,\dots,c_n$ such that $\sem{c_k}=f_k$ for all $k\in\{1,\dots,n\}$. Then, we can compose all these circuits together as specified in the normal form of $m\to n$ causal circuits (Definition~\ref{def:nf-m-n}) which mimics the disintegration of $f$, thereby obtaining a circuit $c$ such that $\sem{c}=f$, as we wanted. 
  \end{proof}
\end{proposition}

As an immediate consequence of Theorem~\ref{thm:completeness-m-to-n} and Proposition~\ref{prop:image-fstoch2}, we can state the following.

\begin{corollary}
  The theory given by the generators in~\eqref{eq:bool-syntax}-\eqref{eq:causal-syntax} and the equations of Fig.~\ref{fig:eqs} is a presentation of $\fStoch_\Bool$.
\end{corollary}

\section{Adding observations}
\label{sec:observe}

The power and usefulness of probabilistic programming languages comes from their ability to   specify distributions conditional on observations defined by the program. This is known as \emph{inference}. As we saw, for this purpose, \textsc{Dice} has $\texttt{observe}\ x$ statements, which constrain the distribution encoded by a program to only those assignments for which the variable $x$ is \texttt{true} (equal to $1$). As explained in Section~\ref{sec:syntax-semantics}, $\Bmult$ plays a similar role in our diagrammatic calculus: it allows us to condition on the value of its two input wires to be equal. 

\subsection{Equational theory}
\label{sec:conditioning-axioms}

We now extend the axiomatisation of Section~\ref{sec:axiomatisation-caus-circ} to circuits, including those containing explicit conditioning with $\Bmult$. The additional axioms are given in Fig.~\ref{fig:conditioning-axioms}.
\begin{itemize}
\item The first line (F1-F3) make $\Bmult$ and $\Flip{1/2}$ into a commutative monoid. Note that the unitality (F2) is only true when we quotient subdistributions by a global multiplicative factor; had we given our semantics in terms of plain subdistributions (without quotient), then these two equalities would not hold.
\item The second line (F4-F6) tell us that $\Bcomult,\Bcounit,\Bmult,\Flip{1/2}$ form a \emph{special Frobenius algebra}. 
\item Axiom F7 encodes the action of conditioning on a distribution over two variables. If we call $x_1,x_2$ these two variables we can see that the diagram on the right gives this distribution as the product of a Bernoulli with parameter $p_0$ for $x_1$ with the distribution of $x_2$ conditional on $x_1$, given by two Bernoulli with parameters $p_1$ and $p_2$. Clearly, these are both true with probability $p_0p_1$ and both false with probability $(1-p_0)(1-p_2)$. Since $\Bmult$ has the effect of constraining $x_1$ and $x_2$ to be equal, the result is a new Bernoulli distribution with parameter $r$, whose denominator is the normalisation factor $p_0p_1+ (1-p_0)(1-p_2)$ and whose numerator is the probability that $x_1$ and $x_2$ are both true. 
\item The final axiom (F8) deals with failure, that is, with the special circuit $\Flip{\bot}:=
\InputIfFileExists{fail.tikz}{}{\input{./tikz/fail.tikz}}
$. This constrains $0$ to be equal to $1$, which is unsatisfiable. Any two circuits that contains failure are equal semantically to the zero subdistribution. As it turns out, axiom F8 is enough to derive this more general fact (in the presence of the other axioms).
\end{itemize}
\begin{figure}
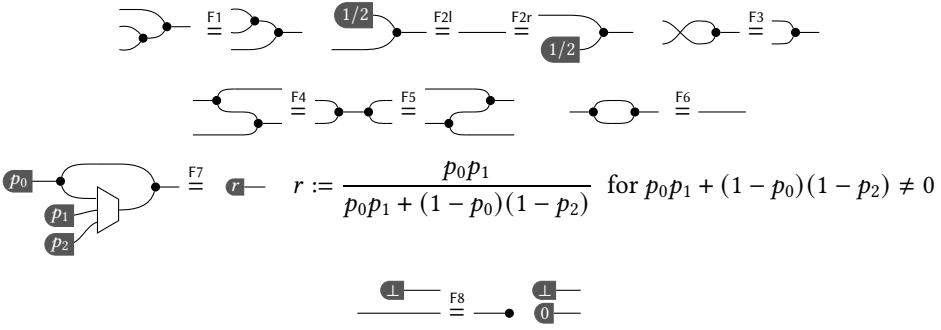

    \centering
    \begin{gather*}
      
\InputIfFileExists{cocopy-associative.tikz}{}{\input{./tikz/cocopy-associative.tikz}}
 \;\myeq{F1} 
\InputIfFileExists{cocopy-associative-1.tikz}{}{\input{./tikz/cocopy-associative-1.tikz}}
 \quad 
\InputIfFileExists{cocopy-unital-left.tikz}{}{\input{./tikz/cocopy-unital-left.tikz}}
 \myeq{F2l} \idone\myeq{F2r} 
\InputIfFileExists{cocopy-unital-right.tikz}{}{\input{./tikz/cocopy-unital-right.tikz}}
\quad  
\InputIfFileExists{cocopy-commutative.tikz}{}{\input{./tikz/cocopy-commutative.tikz}}
 \myeq{F3} \Bmult
    \end{gather*}
        \begin{gather*}
   
\InputIfFileExists{frobenius-s.tikz}{}{\input{./tikz/frobenius-s.tikz}}
 \myeq{F4} 
\InputIfFileExists{frobenius.tikz}{}{\input{./tikz/frobenius.tikz}}
\myeq{F5} 
\InputIfFileExists{frobenius-z.tikz}{}{\input{./tikz/frobenius-z.tikz}}
\qquad 
\InputIfFileExists{bcomult-bmult.tikz}{}{\input{./tikz/bcomult-bmult.tikz}}
\;\myeq{F6}\idone
    \end{gather*}
    \begin{gather*}
   
\InputIfFileExists{cond.tikz}{}{\input{./tikz/cond.tikz}}
 \myeq{F7} \;\;\Flip{r}\quad r:= \frac{p_0p_1}{p_0p_1+(1-p_0)(1-p_2)} \;\text{ for }  p_0p_1+(1-p_0)(1-p_2)\neq 0
    \end{gather*}
    \begin{gather*}
   
\InputIfFileExists{fail-id.tikz}{}{\input{./tikz/fail-id.tikz}}
 \myeq{F8} 
\InputIfFileExists{fail-0.tikz}{}{\input{./tikz/fail-0.tikz}}

    \end{gather*}
\caption{Axioms for conditioning.}\label{fig:conditioning-axioms}
\end{figure} 

\begin{theorem}[Soundness]
For any two circuits $c,d\from m\to n$, if $c=d$ then $\condsem{c}=\condsem{d}$.
\end{theorem}
\begin{proof}
Many of the axioms are straightforward and hold already at the level of the subdistributional semantics given by $\sem{\cdot}$. We focus on those that require the semantics in $\fProjStoch$ given by $\condsem{\cdot}$. Recall that we use $[f]$ to denote the equivalence class of a substochastic map/subdistribution $f$ under the equivalence relation $\propto$.
\begin{description}
\item[(F2)] We have
\[\condsem{
\InputIfFileExists{cocopy-unital-left.tikz}{}{\input{./tikz/cocopy-unital-left.tikz}}
}(x) = \left[\frac{1}{2}\ket{x}\right]= [\ket{x}] = \condsem{\idone}(x) \]
since the first equality holds uniformly for all $x\in \Bool$.
\item[(F7)] We have
\begin{align*}
\condsem{
\InputIfFileExists{nf-0-2.tikz}{}{\input{./tikz/nf-0-2.tikz}}
} = [p_0p_1\ket{11} + p_0(1-p_1)\ket{10} + (1-p_0)p_2\ket{01} + (1-p_0)(1-p_2)\ket{00}]
\end{align*}
and, since only the first and the last summands satisfy the constraint imposed by $\Bmult$:
\begin{align*}
\condsem{
\InputIfFileExists{cond.tikz}{}{\input{./tikz/cond.tikz}}
} &= [p_0p_1\ket{1} + (1-p_0)(1-p_2)\ket{0}] 
\\
&= \left[\frac{p_0p_1}{p_0p_1+(1-p_0)(1-p_2)} \ket{11} + \frac{(1-p_0)(1-p_2)}{p_0p_1+(1-p_0)(1-p_2)} \ket{00}\right] = \condsem{\Flip{r}}
\end{align*}
where the last equality renormalises the previous subdistribution, which is only possible when  $p_0p_1+(1-p_0)(1-p_2)\neq 0$.
\item[(F8)] We have $\condsem{\Flip{\bot}}:=\condsem{
\InputIfFileExists{fail.tikz}{}{\input{./tikz/fail.tikz}}
} = 0$. Thus any two circuits composed in parallel with $\Flip{\bot}$ have the zero subdistribution as denotation, and in particular, axiom F8 is sound.
\end{description}
\end{proof}

\subsection{Completeness}
\label{sec:conditioning-completeness}

We now show the completeness of the proposed equational theory.
\begin{itemize}
\item First, we will reduce the problem to the case of circuits without input wires, that is, circuits of type $0\to n$. This is possible because the Frobenius algebra structure of $\Bcomult, \Bcounit$, $\Bmult,\Flip{1/2}$ with which we can bend wires using $
\InputIfFileExists{cup-frob.tikz}{}{\input{./tikz/cup-frob.tikz}}
$ and $
\InputIfFileExists{cap-frob.tikz}{}{\input{./tikz/cap-frob.tikz}}
$. This allows us to transform any $m\to n$ circuits into one of type $0\to m+n$ and vice-versa; moreover, any equality that we can apply to the former also applies to the latter (Proposition~\ref{prop:compact-closed}). In categorical terms, we say that our semantics, the SMC $\fProjStoch$, is \emph{compact-closed}. 
\item Then, the main idea of the completeness proof itself is to remove all explicit conditioning ($\Bmult$) from $0\to n$ circuits, in order to reduce the completeness of the whole syntax to the causal fragment we have studied in previous sections. In this process, we will make crucial use of the normal for causal circuits. The main difference with the causal case is that conditioning introduces the possibility of failure: circuits that contain the unsatisfiable constraint $\Flip{\bot}:=
\InputIfFileExists{fail.tikz}{}{\input{./tikz/fail.tikz}}
$ and whose semantics is the zero subdistribution. As we will see we have to deal with this special case separately, since any two circuits containing failure are equal (Lemma~\ref{lem:fail}).
\end{itemize}

\begin{proposition}[Compactness]\label{prop:compact-closed}
There is a one-to-one mapping $\Bend{(\cdot)}$ between diagrams of type $m\to n$ and those of type $0\to m+n$; in addition, for any $c;d\from m \to n$, $c =d$ iff $\Bend{c}=\Bend{d}$.
\end{proposition}
\begin{proof}
This is a standard fact that holds in any \emph{compact-closed category}~\cite{kelly1980}. The mapping $\Bend{(\cdot)}$ is defined on a diagram $c$ by
\[
\Bend{\left(\probcircuit{c}{m}{n}\right)} := 
\InputIfFileExists{c-bend.tikz}{}{\input{./tikz/c-bend.tikz}}

\]
The inverse mapping is then given by:
\[

\InputIfFileExists{c-unbend.tikz}{}{\input{./tikz/c-unbend.tikz}}

\]
That these two transformations are inverses of each other is a consequence of axioms F2 and F4-F5:
\[

\InputIfFileExists{compact-proof.tikz}{}{\input{./tikz/compact-proof.tikz}}
\;\myeq{F5}\;
\InputIfFileExists{compact-proof-1.tikz}{}{\input{./tikz/compact-proof-1.tikz}}
\;\myeq{F2}\;
\InputIfFileExists{compact-proof-2.tikz}{}{\input{./tikz/compact-proof-2.tikz}}
\;\myeq{A2}\;\probcircuit{c}{m}{n}
\]
Finally, the second part of the statement is immediate, since $\Bend{(\cdot)}$ is defined by composing with generators of our syntax, and therefore, any equality that we can apply to show that $c=d$ can be applied to show that $\Bend{c}=\Bend{d}$ and vice-versa.
\end{proof}

\begin{theorem}[Completeness]
\label{thm:conditioning-completeness}
For any two circuits $c,d\from m\to n$ with conditioning, $\sem{c}=\sem{d}$ iff $c=d$.
\end{theorem}
\begin{proof}
By compactness (Proposition~\ref{prop:compact-closed}), it is enough to prove the statement for $c,d\from 0 \to n$. We want to show that we can always eliminate all conditioning nodes $\Bmult$ to rewrite any diagram $c\from 0\to n$ into normal form. For this, it is enough to show that any diagram $c\from 0\to n$ composed (on the right, necessarily) with $\Bmult$ is equal to some conditioning-free circuit.

First, we can assume wlog that $\Bmult$ is plugged on the first two wires of $c$; if it was not, we could simply reorder them and normalise the resulting diagram first. Moreover, by Theorem~\ref{thm:completeness-m-to-n}, there exists circuits $c_0$ and $c_1$ in normal form such that:
\begin{align*}

\InputIfFileExists{c-bmultxid.tikz}{}{\input{./tikz/c-bmultxid.tikz}}
 & \;=\;
\InputIfFileExists{c-disint-bmultxid.tikz}{}{\input{./tikz/c-disint-bmultxid.tikz}}
\; =\;  
\InputIfFileExists{c-disint-2-bmultxid.tikz}{}{\input{./tikz/c-disint-2-bmultxid.tikz}}

\\
& =\; 
\InputIfFileExists{c-disint-2-bmultxid-2.tikz}{}{\input{./tikz/c-disint-2-bmultxid-2.tikz}}
\;\myeq{A1}\; 
\InputIfFileExists{c-disint-2-bmultxid-3.tikz}{}{\input{./tikz/c-disint-2-bmultxid-3.tikz}}

\\
& \myeq{Fr}\; 
\InputIfFileExists{c-disint-2-bmultxid-4.tikz}{}{\input{./tikz/c-disint-2-bmultxid-4.tikz}}
 \;\myeq{A3}\; 
\InputIfFileExists{c-disint-2-bmultxid-5.tikz}{}{\input{./tikz/c-disint-2-bmultxid-5.tikz}}

\\
& \myeq{Fr}\;
\InputIfFileExists{c-disint-2-bmultxid-6.tikz}{}{\input{./tikz/c-disint-2-bmultxid-6.tikz}}
\;\myeq{A3}\;
\InputIfFileExists{c-disint-2-bmultxid-7.tikz}{}{\input{./tikz/c-disint-2-bmultxid-7.tikz}}

\end{align*}
Since $c_0,c_1$ are themselves in normal form, there exists weights $p_0,q_0,q_1$ such that
\begin{align*}

\InputIfFileExists{c-disint-2-bmultxid-7.tikz}{}{\input{./tikz/c-disint-2-bmultxid-7.tikz}}
\;=\; 
\InputIfFileExists{c-disint-2-bmultxid-8.tikz}{}{\input{./tikz/c-disint-2-bmultxid-8.tikz}}

\end{align*}
Now there are two cases to consider. 
\begin{itemize}
\item First, if $p_0q_0+(1-p_0)(1-q_1)\neq 0$, we have
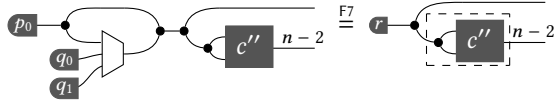
\begin{align*}
 
\InputIfFileExists{c-disint-2-bmultxid-8.tikz}{}{\input{./tikz/c-disint-2-bmultxid-8.tikz}}
\;\myeq{F7}\; 
\InputIfFileExists{c-disint-2-bmultxid-fin.tikz}{}{\input{./tikz/c-disint-2-bmultxid-fin.tikz}}

\end{align*}
where $r:= \frac{p_0q_0}{p_0q_0+(1-p_0)(1-q_1)}$. Notice that the sub-circuit in the dashed box above is in normal form, so the final circuit is also in normal form, as we wanted to show.
\item Second, if $p_0q_0+(1-p_0)(1-q_1)=0$, then we either have $p_0=1$ and $q_0=0$, or $p_0=0$ and $q_1=1$. Intuitively, in these two cases, the conditioning constraint imposed by $\Bmult$ is unsatisfiable, and we get:
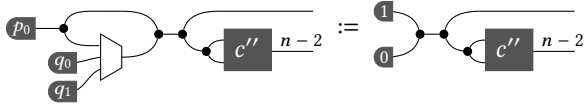
\begin{align*}
 
\InputIfFileExists{c-disint-2-bmultxid-8.tikz}{}{\input{./tikz/c-disint-2-bmultxid-8.tikz}}
\;:=\; 
\InputIfFileExists{c-disint-2-bmultxid-fail.tikz}{}{\input{./tikz/c-disint-2-bmultxid-fail.tikz}}

\end{align*}
Note that this is true in both cases, by commutativity of $\Bmult$. Then,
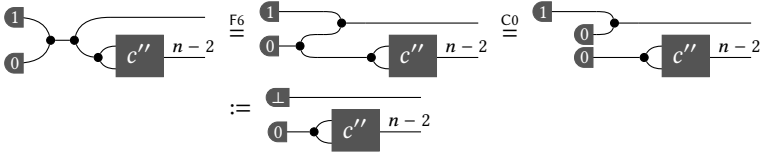
\begin{align*}
 
\InputIfFileExists{c-disint-2-bmultxid-fail.tikz}{}{\input{./tikz/c-disint-2-bmultxid-fail.tikz}}
\;&\myeq{F6}\; 
\InputIfFileExists{c-disint-2-bmultxid-fail-1.tikz}{}{\input{./tikz/c-disint-2-bmultxid-fail-1.tikz}}
 \;\myeq{C0}\; 
\InputIfFileExists{c-disint-2-bmultxid-fail-2.tikz}{}{\input{./tikz/c-disint-2-bmultxid-fail-2.tikz}}
\;
 \\
 &:=\;
\InputIfFileExists{c-disint-2-bmultxid-fail-3.tikz}{}{\input{./tikz/c-disint-2-bmultxid-fail-3.tikz}}

\end{align*}
Since any two circuits containing $\Flip{\bot}$ are equal by Lemma~\ref{lem:fail}, the last circuit is equal to one in normal form.
\end{itemize}
Thus, any $c\from 0\to n$ is equal to a circuit in normal form and the same reasoning as in the proof of Theorem~\ref{thm:completeness-m-to-n} shows that any two semantically equivalent circuits $0\to n$ are equal to the same circuit in normal form, which concludes the proof.
\end{proof}
The following lemma states that any two circuits which fail are equal. 
\begin{lemma}[Failure]
\label{lem:fail}
For any two circuits $c,d\from m \to n$, we have
\[
\InputIfFileExists{fail-c.tikz}{}{\input{./tikz/fail-c.tikz}}
\;=\;
\InputIfFileExists{fail-d.tikz}{}{\input{./tikz/fail-d.tikz}}
\]
\end{lemma}
\begin{proof}
To prove the main statement, we need only show that, for any $c\from m\to n$, the following equality holds:
\begin{equation}
\label{eq:fail}

\InputIfFileExists{fail-c.tikz}{}{\input{./tikz/fail-c.tikz}}
\;=\; 
\InputIfFileExists{fail-0-m-n.tikz}{}{\input{./tikz/fail-0-m-n.tikz}}

\end{equation}
We prove this by structural induction on $c$. First, for the generators as base cases:
\begin{itemize}
\item For $\Bcomult$, we have
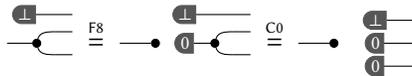
\[
\InputIfFileExists{fail-bcomult.tikz}{}{\input{./tikz/fail-bcomult.tikz}}
\;\myeq{F8}\;
\InputIfFileExists{fail-bcomult-1.tikz}{}{\input{./tikz/fail-bcomult-1.tikz}}
\;\myeq{C0}\;
\InputIfFileExists{fail-bcomult-2.tikz}{}{\input{./tikz/fail-bcomult-2.tikz}}
\]
\item The case of $\Bcounit$ is trivial.
\item For $\Andgate$, we have
\[
\InputIfFileExists{fail-and.tikz}{}{\input{./tikz/fail-and.tikz}}
\;=\;
\InputIfFileExists{fail-and-1.tikz}{}{\input{./tikz/fail-and-1.tikz}}
\;\myeq{F8}\;
\InputIfFileExists{fail-and-2.tikz}{}{\input{./tikz/fail-and-2.tikz}}
\;\myeq{D1}\;
\InputIfFileExists{fail-and-3.tikz}{}{\input{./tikz/fail-and-3.tikz}}
\]
\item For $\Notgate$, we have
\[
\InputIfFileExists{fail-not.tikz}{}{\input{./tikz/fail-not.tikz}}
\;=\;
\InputIfFileExists{fail-not-1.tikz}{}{\input{./tikz/fail-not-1.tikz}}
\;\myeq{F8}\;
\InputIfFileExists{fail-not-2.tikz}{}{\input{./tikz/fail-not-2.tikz}}
\;\myeq{D2}\;
\InputIfFileExists{fail-0.tikz}{}{\input{./tikz/fail-0.tikz}}
\]
\item For $\Flip{p}$, we have
\[
\begin{tikzpicture}
	\begin{pgfonlayer}{nodelayer}
		\node [style=flip] (1) at (0.75, 0.75) {$\bot$};
		\node [style=none] (3) at (2.25, 0.75) {};
		\node [style=flip] (5) at (0.75, -0.25) {$p$};
		\node [style=none] (6) at (2.25, -0.25) {};
	\end{pgfonlayer}
	\begin{pgfonlayer}{edgelayer}
		\draw (1) to (3.center);
		\draw (5) to (6.center);
	\end{pgfonlayer}
\end{tikzpicture}
}
\;=\;
\begin{tikzpicture}
	\begin{pgfonlayer}{nodelayer}
		\node [style=flip] (1) at (0.75, 0.75) {$\bot$};
		\node [style=none] (3) at (2.25, 0.75) {};
		\node [style=flip] (5) at (-1, -0.25) {$p$};
		\node [style=none] (6) at (2.25, -0.25) {};
	\end{pgfonlayer}
	\begin{pgfonlayer}{edgelayer}
		\draw (1) to (3.center);
		\draw (5) to (6.center);
	\end{pgfonlayer}
\end{tikzpicture}
}
\;\myeq{F8}\;
\InputIfFileExists{fail-flip-2.tikz}{}{\input{./tikz/fail-flip-2.tikz}}
\;\myeq{D3}\;
\InputIfFileExists{fail-0.tikz}{}{\input{./tikz/fail-0.tikz}}
\]
\item For $\Bmult$, we have
\begin{align*}

\InputIfFileExists{fail-bmult.tikz}{}{\input{./tikz/fail-bmult.tikz}}
&\;=\;
\InputIfFileExists{fail-bmult-1.tikz}{}{\input{./tikz/fail-bmult-1.tikz}}
\;\myeq{F8}\;
\InputIfFileExists{fail-bmult-2.tikz}{}{\input{./tikz/fail-bmult-2.tikz}}
\;=\;
\InputIfFileExists{fail-bmult-3.tikz}{}{\input{./tikz/fail-bmult-3.tikz}}
\\
& =\; 
\InputIfFileExists{fail-bmult-4.tikz}{}{\input{./tikz/fail-bmult-4.tikz}}
\;\myeq{C0}\; 
\InputIfFileExists{fail-bmult-5.tikz}{}{\input{./tikz/fail-bmult-5.tikz}}
\;\myeq{F6}\; 
\InputIfFileExists{fail-bmult-6.tikz}{}{\input{./tikz/fail-bmult-6.tikz}}

\end{align*}
\end{itemize}
The two inductive case of sequential and parallel composition are equally straightforward.
\begin{itemize}
\item For sequential composition, assume that circuits $c\from \ell\to m$ and $d\from m \to n$ satisfy~\eqref{eq:fail}; then
\begin{align*}

\InputIfFileExists{fail-c-d.tikz}{}{\input{./tikz/fail-c-d.tikz}}
&\;=\;
\InputIfFileExists{fail-c-d-1.tikz}{}{\input{./tikz/fail-c-d-1.tikz}}
\;=\;
\InputIfFileExists{fail-c-d-2.tikz}{}{\input{./tikz/fail-c-d-2.tikz}}

\\
&\;=\;
\InputIfFileExists{fail-c-d-3.tikz}{}{\input{./tikz/fail-c-d-3.tikz}}
\;\myeq{D3}\;
\InputIfFileExists{fail-c-d-4.tikz}{}{\input{./tikz/fail-c-d-4.tikz}}

\end{align*}
\item For parallel composition, assume that $c_1\from m_1\to n_1$ and $c_2\from m_2\to n_2$ satisfy~\eqref{eq:fail}; then
\begin{align*}

\InputIfFileExists{fail-c1xc2.tikz}{}{\input{./tikz/fail-c1xc2.tikz}}
&\;=\; 
\InputIfFileExists{fail-c1xc2-1.tikz}{}{\input{./tikz/fail-c1xc2-1.tikz}}
 \;=\;
\InputIfFileExists{fail-c1xc2-2.tikz}{}{\input{./tikz/fail-c1xc2-2.tikz}}

\\
& = \; 
\InputIfFileExists{fail-c1xc2-3.tikz}{}{\input{./tikz/fail-c1xc2-3.tikz}}
\;=\; 
\InputIfFileExists{fail-c1xc2-4.tikz}{}{\input{./tikz/fail-c1xc2-4.tikz}}

\end{align*}
\end{itemize}
\end{proof}
From these axioms, we can derive the action of $\Bmult$ on $\Flip{p}$ generators: it multiplies the parameters and re-normalises the resulting subdistribution, as stated in the following lemma.
\begin{lemma}
\label{lem:mult}
The following identity is derivable for $p,q\in(0,1)$:
\[
\InputIfFileExists{mult.tikz}{}{\input{./tikz/mult.tikz}}
\;=\;\Flip{r}\quad \text{where } r:= \frac{pq}{pq + (1-p)(1-q)}\]
\end{lemma}
\begin{proof}
\begin{align*}

\InputIfFileExists{mult.tikz}{}{\input{./tikz/mult.tikz}}
\;\myeq{A2}\;
\InputIfFileExists{mult-1.tikz}{}{\input{./tikz/mult-1.tikz}}
\;\myeq{$(\star)$}\;
\InputIfFileExists{mult-fin.tikz}{}{\input{./tikz/mult-fin.tikz}}
\;\myeq{F7}\;\Flip{r}
\end{align*}
with $r$ defined as in the statement of the lemma, and the equality labelled ($\star$) is derivable as follows: first, some simple Boolean algebra manipulations give us
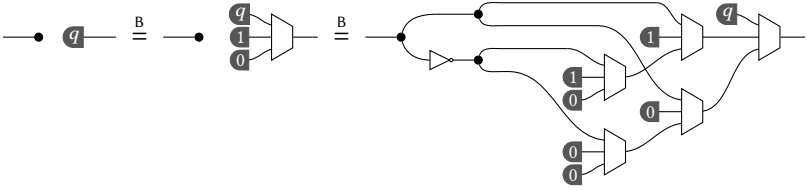
\begin{align*}

\InputIfFileExists{del-flip-q.tikz}{}{\input{./tikz/del-flip-q.tikz}}
\;\myeq{B}\;
\InputIfFileExists{del-flip-q-1.tikz}{}{\input{./tikz/del-flip-q-1.tikz}}
\;\myeq{B}\; 
\InputIfFileExists{del-flip-q-2.tikz}{}{\input{./tikz/del-flip-q-2.tikz}}

\end{align*}
Then, we apply E5 repeatedly:
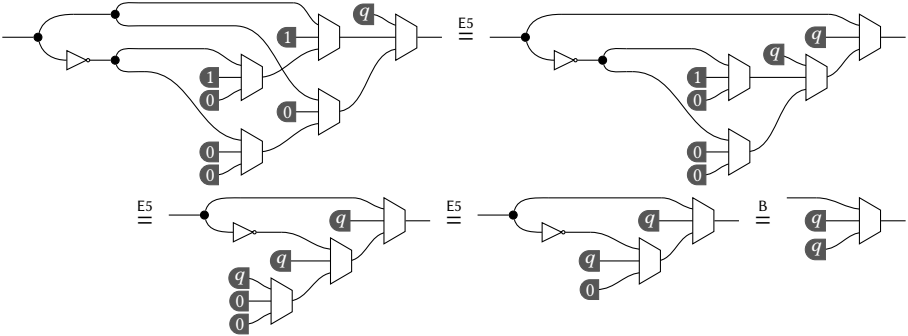
\begin{align*}

\InputIfFileExists{del-flip-q-2.tikz}{}{\input{./tikz/del-flip-q-2.tikz}}
\;\myeq{E5}\;
\InputIfFileExists{del-flip-q-3.tikz}{}{\input{./tikz/del-flip-q-3.tikz}}

\\
\quad\myeq{E5}\;
\InputIfFileExists{del-flip-q-4.tikz}{}{\input{./tikz/del-flip-q-4.tikz}}
\;\myeq{E5}\;
\InputIfFileExists{del-flip-q-5.tikz}{}{\input{./tikz/del-flip-q-5.tikz}}
\;\myeq{B}\;
\InputIfFileExists{if-flip-qxflip-q.tikz}{}{\input{./tikz/if-flip-qxflip-q.tikz}}

\end{align*}
where the last equality is another simple Boolean algebra identity.
\end{proof}
We can now prove the correctness of Von Neumann's trick to simulate a fair coin from a biased one (\emph{cf.} Example~\ref{ex:von-neumann}).
\begin{example}[Von Neumann's trick, proof of correctness]
Assume $p\in (0,1)$. We have
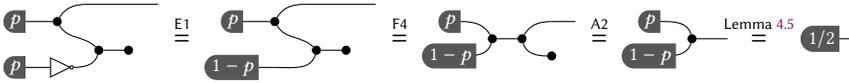
\begin{align*}

\InputIfFileExists{von-neumann.tikz}{}{\input{./tikz/von-neumann.tikz}}
\;\myeq{E1}\;
\InputIfFileExists{von-neumann-1.tikz}{}{\input{./tikz/von-neumann-1.tikz}}
\;\myeq{F4}\;
\InputIfFileExists{von-neumann-2.tikz}{}{\input{./tikz/von-neumann-2.tikz}}

\;\myeq{A2}\;
\InputIfFileExists{von-neumann-3.tikz}{}{\input{./tikz/von-neumann-3.tikz}}
\;\;\myeq{Lemma~\ref{lem:mult}}\quad\Flip{1/2}
\end{align*}
since
\[\frac{p(1-p)}{p(1-p) + (1-p)(1-(1-p))}=\frac{p(1-p)}{p(1-p) + (1-p)p} = \frac{p(1-p)}{2p(1-p)}=\frac{1}{2}\]
\end{example}

\section{Conclusion and Future Work}
\label{sec:conclusion}

In summary, we presented a complete axiomatisation of discrete probabilistic program equivalence. We were able to achieve this result by translating a conventional PL into a calculus of string diagrams, representing Boolean circuits extended with primitives denoting Bernoulli distributions and explicit conditioning, which we equipped with a compositional semantics in terms of subdistributions (up to some re-normalisation factor).  Then, we gave a complete equational theory for circuits without conditioning, that is, a presentation of the category of stochastic maps between sets that are powers of the Booleans. Finally, we showed how to extend the preceding axiomatisation to all circuits---and therefore all programs---including those that contain explicit conditioning. 

This opens up several directions for future work. Firstly, we believe that axiomatisations of alternative semantics for conditioning can be achieved with a few simple modifications. In particular, the semantics of conditioning in terms of plain subdistributions (without re-normalisation) is within reach. Secondly, we would like to axiomatise KL-divergence between probabilistic circuits/programs using a \emph{quantitative} equational theory~\cite{bacci2020quantitative,perrone2023markov}.  This would require us to replace plain equalities with equalities of the form $c\equiv_\varepsilon d$, to represent the fact that the KL-divergence of $\sem{c}$ relative to $\sem{d}$ is less than $\varepsilon$. Interesting new identities can be expressed in this framework, such as
\begin{center}
  
\InputIfFileExists{kl-d-axiom.tikz}{}{\input{./tikz/kl-d-axiom.tikz}}
 \quad where $\varepsilon \geq p \ln\left(\dfrac{p}{p^2}\right)+(1-p)\ln\left(\dfrac{1-p}{(1-p)^2}\right)$.
\end{center}
Such a quantitative identity measures the information loss of assuming that some phenomenon can be modelled by a single coin flip when the true distribution consists of two independent coin tosses. The same approach could also be used to axiomatise other quantitative measures of program divergence, such as the total variation distance~\cite{perrone2023markov}. 

\begin{acks}
The authors would like to thank Wojciech R{\'o}{\.z}owski and Dario Stein for many helpful discussions. 
\end{acks}

\bibliographystyle{ACM-Reference-Format}
\bibliography{refs}

\appendix

\end{document}